\def\DRAFT{}

\documentclass[11pt]{article}

\usepackage[T1]{fontenc}
\usepackage{lmodern, microtype}
\usepackage[left=1.3in, right=1.3in, top=1.25in, bottom=1.25in]{geometry}
\usepackage[small]{titlesec}
\usepackage{mathtools,comment,paralist}
\usepackage{epigraph}

\usepackage{amssymb,amsmath,amsthm,fancyhdr,bbm}

\usepackage{xcolor}
\usepackage[normalem]{ulem} 
\usepackage{hyperref}
\hypersetup{
  colorlinks=true,
  linkcolor=blue,
  citecolor=red
}
\usepackage{natbib}
\usepackage{subcaption}
\setcitestyle{authoryear,open={(},close={)}}
\makeatletter

\newtheorem{theorem}{Theorem}
\newtheorem{lemma}{Lemma}
\newtheorem*{lemma*}{Lemma}

\newtheorem*{axiom*}{Axiom}
\newtheorem{proposition}{Proposition}
\newtheorem{claim}{Claim}
\newtheorem{corollary}{Corollary}

\newtheorem*{theorem*}{Theorem}

\newtheorem{definition}{Definition}
\newtheorem*{definition*}{Definition}
\usepackage{graphicx}
\usepackage{pstricks, enumerate, pst-node, pst-text, pst-plot}

\theoremstyle{remark}
\newtheorem{remark}{Remark}

\DeclareMathOperator\supp{supp}
\DeclareMathOperator\g{>_{FOSD}}

\newcommand{\ra}{\Rightarrow}

\newcommand{\R}{\mathbb{R}}
\newcommand{\Q}{\mathbb{Q}}
\newcommand{\N}{\mathbb{N}}

\newcommand{\E}{\mathbb{E}}

\usepackage{accents}
\newcommand{\ubar}[1]{\underaccent{\bar}{#1}}

\def\dd{\mathrm{d}}
\def\fin{\mathrm{finite}}
\def\ee{\mathrm{e}}
\def\LQRE{\mathrm{LQRE}}
\def\LP{\mathrm{LQRE_{\lambda \Phi}}}
\def\Nash{\mathrm{Nash}}
\newcommand\norm[1]{\left\Vert#1\right\Vert}

\def\cA{\mathcal{A}}

\def\cO{\mathcal{O}}

\DeclarePairedDelimiter\ceil{\lceil}{\rceil}
\DeclarePairedDelimiter\floor{\lfloor}{\rfloor}

\DeclareMathOperator*{\argmax}{\arg\max}
\DeclareMathOperator*{\argmin}{\arg\min}
\DeclareMathOperator*{\linf}{lim\!\,inf}
\DeclareMathOperator*{\lsup}{lim\!\,sup}

\ifdefined\DRAFT

\definecolor{ForestGreen}{rgb}{.13,.54,.13}
\definecolor{violet}{cmyk}{0.79,0.88,0,0}
\definecolor{darkmagenta}{rgb}{0.55, 0.0, 0.55}

 \newcommand{\fed}[1]{{\color{red}{(\textbf{Fedor:} #1)}}}
 \newcommand{\omer}[1]{{\color{blue}{(\textbf{Omer:} #1)}}}
 \newcommand{\benpo}[1]{{\color{darkmagenta}{(\textbf{Ben \& Po:} #1)}}}

\definecolor{darkBrickRed}{rgb}{.70,.13,.16}

\else
\newcommand{\fed}[1]{}

\newcommand{\omer}[1]{}
\newcommand{\benpo}[1]{}
\fi

\definecolor{redPart}{rgb}{1.0, 0.0, 0.0}
\definecolor{greenPart}{rgb}{0.0, 1, 0.0}
\definecolor{bluePart}{rgb}{0.0, 0.0, 1.0}
\definecolor{yellowPart}{rgb}{1, 1, 0}
\definecolor{magentaPart}{rgb}{1, 0, 1}
\definecolor{cyanPart}{rgb}{0.0, 1, 1}

\title{Monotonicity and Bracketing in Games\footnote{
The paper benefited from 
our colleagues'
suggestions and comments. We are grateful
to the audience at the Cowles Foundation conference and Stony Brook, as well as Marina Agranov, James Bland, Daniel Chen, Drew Fudenberg, Elliot Lipnowski, Fabio Maccheroni, George Mailath, Yusufcan Masatlioglu, Ellen Muir, Roger Myerson, Tom Palfrey, Luciano Pomatto, Andrew Postlewaite, Doron Ravid, Joseph Root, Larry Samuelson, Eran Shmaya, Ran Shorrer, Charles Sprenger, Chloe Tergiman, William Thomson, and Leeat Yariv.
}}

\author{Fedor Sandomirskiy\thanks{Princeton University. Email: fsandomi@princeton.edu} \and
        Po Hyun Sung\thanks{Caltech. Email: psung@caltech.edu} \and
        Omer Tamuz\thanks{Caltech. Email: tamuz@caltech.edu. Omer Tamuz was supported by a BSF award (\#2018397), and a National Science Foundation CAREER award (DMS-1944153).} \and
        Ben Wincelberg\thanks{Caltech. Email: bwincelb@caltech.edu}}

\date{\today}

\begin{document}

\maketitle
\begin{abstract}

We study solution concepts for normal-form games. 
We obtain a characterization of Nash equilibria and logit quantal response equilibria, as well as generalizations capturing non-expected utility. Our axioms reflect that players are responsive to payoffs induced by the play of others and, whenever several games are played simultaneously, 
players may consider each separately.

\end{abstract}
\vskip 1cm


\section{Introduction}

Consider a player who plays both poker and rock-paper-scissors at a game night. Since these are separate games, a simple modeling choice for describing her decision-making is to assume that she chooses her strategies independently in the two games. If we use Nash equilibrium as our solution concept in analyzing this composite game, then we are guaranteed to have a solution that is of this form  (alongside other equilibria, in which strategies are correlated). In this paper, we explore solution concepts that likewise do not rule out independent play in composite games.

We say that such solution concepts \emph{permit bracketing}. This allows players to treat each game separately, simplifying their decision-making. Bracketing has been extensively documented in decision problems \citep*{read1999choice}, and also in strategic settings: in \cite*{bland2019many}, most subjects bracket when playing a composite game even when doing so is suboptimal. In our case, however, bracketing is consistent with rationality, as we only apply it in separable settings. It also excludes no behavior, because we do not require players to bracket in every solution. Accordingly,  bracketing can be viewed as a richness assumption on the set of solutions in separable environments: it allows the analyst to ignore irrelevant factors when modeling the situation of interest, without requiring every solution to take a separable form.

By itself, our bracketing assumption is not very strong 
as it does not rule out other behavior and only applies to composite games.
We augment bracketing with  monotonicity/rationality assumptions that capture the idea that players anticipate the behavior of others and respond to the strategies of others, playing actions that yield higher payoffs more often. Our main monotonicity axiom is \emph{distribution-monotonicity}, which posits that players are less likely to play actions that yield stochastically-dominated payoff distributions.

A stronger axiom, which we study to build intuition, is \emph{expectation-monotonicity}, which posits that players are less likely to play actions that yield lower expected payoffs. In our first result, we consider solution concepts that permit bracketing and satisfy expectation-monotonicity (Theorem~\ref{th_Nash QRE}). We demonstrate that these axioms characterize Nash equilibrium and logit quantal response equilibrium (LQRE), providing a novel unifying perspective on these established solution concepts. 
The broader family of quantal response equilibria (QRE) has been criticized for being overly permissive \citep*{haile2008empirical}.\footnote{This critique applies to QRE with unrestricted quantal response functions \citep*{mckelvey1995quantal}. Restrictions have thus been imposed on quantal response functions to obtain 
falsifiable restrictions on behavior; see \cite*{goeree2005regular}.}
Our axioms single out the one-parameter logit subfamily, offering a
justification for the use of LQRE and elevating it beyond its typical role of a computationally convenient choice.

In our main result, we consider the more general case of solution concepts that permit bracketing and satisfy distribution-monotonicity. This leads to a characterization of new solution concepts that incorporate non-expected-utility in equilibrium behavior, and yet retain much of the conceptual appeal of Nash equilibria and LQRE (Theorem~\ref{th_sre}). In these solution concepts, which we call \emph{statistic response equilibria}, players respond to a statistic: a single number assigned to a distribution of payoffs. Players use the same statistic in every game, and so they behave as if they have a preference over lotteries that applies across various strategic situations. When the statistic is the expectation, the solution concept boils down to Nash or LQRE. More generally, the statistic is a monotone additive statistic, \`a la~\cite*{mu2024monotone}.

The class of statistic response equilibria is broad, and we characterize a parametric subclass that maintains flexibility. Using a scale-invariance axiom, we characterize a three-parameter family of statistic response equilibria, which we expect to be useful in estimating empirical models of games (Theorem~\ref{th_complex}). In these \emph{Min-Max-Mean response equilibria}, players logit best respond to a convex combination of the minimum, maximum, and expectation of the payoff distributions. 
\medskip

In conclusion, our paper develops a conceptual framework that grounds commonly used solution concepts and motivates new generalizations of them.
For theorists, we provide a flexible approach to defining non-expected-utility equilibrium concepts that preserve key properties of classical equilibrium notions. In particular, our axioms imply that behavior is driven by a single preference over lotteries that is stable across games. For empiricists, we offer guidance in selecting models that are based on the established principles of rationality and bracketing. These models are parsimonious, yet flexible enough to accommodate a range of behavior.

\subsection{Related literature}
We contribute to a large body of literature on the axiomatic approach in economic theory.
This methodology has been used extensively to characterize solutions with desirable properties for cooperative games, bargaining,
and mechanism design problems; see surveys by \cite*{moulin1995cooperative}, \cite*{roth2012axiomatic}, and \cite*{thomson2023axiomatics}. In non-cooperative game theory, the axiomatic approach has been
applied primarily toward choosing equilibrium refinements \citep*{harsanyi1988general,norde1996equilibrium, govindan2012axiomatic}.

Our work contributes to a smaller strand of axiomatic literature that aims to characterize solution concepts for non-cooperative games. \cite*{peleg1996consistency} characterize Nash equilibrium and some of its refinements by assuming that players only care about expected payoffs, pick the highest-utility actions in one-player games, and, if some players are replaced with dummies playing their predicted strategies, the behavior of the remaining players remains consistent with the original prediction. This characterization highlights the population-consistency aspect of Nash equilibrium; see also \cite[Chapter 8]{thomson2024consistent} and \cite*{norde1996equilibrium} for further discussion. 

More recently,  \cite*{brandl2024axiomatic} provided a novel axiomatic characterization of Nash equilibria. Their axioms focus on how players' behavior is affected by negligible changes in the strategic environment. We discuss in depth the connection between their axioms and ours in Appendix~\ref{sec_bnb}, where we also explain which of their axioms are satisfied by our statistic response equilibria.
Other axiomatic results include \cite*{brandl2019justifying}, who characterize maximin strategies in zero-sum games, as well as \cite*{babichenko2014axiomatic} and \cite*{voorneveld2019axiomatization}, where pure Nash equilibria are derived using a strategic invariance axiom similar to the one in our Appendix~\ref{app_strategic_eq}.

The key axiom in our approach is  bracketing. This term is often used in a broader sense than in our paper and includes neglect of interactions between related choices. It has been extensively studied in the context of individual decisions, where much of the literature treats  bracketing as a behavioral bias \citep*[see, e.g.,][]{read1999choice,barberis2006individual}. \cite*{bland2019many} experimentally documents bracketing in games, and specifically in composite games.
Our paper belongs to a recent literature offering a rational perspective on  bracketing. \cite*{kHoszegi2020choice} justify  bracketing as rational behavior in a model with costly attention, and \cite*{camara2022computationally} offers a computational complexity justification.

The basic idea behind the proof of Theorem~\ref{th_Nash QRE} originated in our previous paper~\citep*{sandomirskiy2025iid}, where we study single-agent decisions and characterize random coefficients logit. The main axiom in that paper (independence of irrelevant decisions, or IID) stipulates that the prediction for composite decisions should have the same marginals as the prediction for the components. This is different from our bracketing axiom in several aspects. First, the bracketing axiom does not require every prediction to have the same marginals, but only that the solution concept includes one such prediction. Second, in that one prediction we also require independence, which is not part of the IID axiom. Accordingly, the key difficulty in \cite*{sandomirskiy2025iid} is handling correlated decisions, which is achieved via the De Finetti theorem. In this paper, the main difficulty is in the proof of Theorem~\ref{th_sre}, and stems from the endogeneity of randomness.

\cite*{luce1959individual} used the classical independence of irrelevant alternatives (IIA) axiom to characterize multinomial logit choice in decision problems, which is equivalent to LQRE in single-player games. As we explain at the end of \S\ref{sec_3}, our axioms are conceptually and technically different from IIA. While one could expect that LQRE in games could also be characterized using a modification of IIA, we highlight some difficulties in this approach for more than one player.

Our first set of results provides a rationality-based justification for the logit subfamily of quantal response equilibria (QRE).
Introduced by \cite*{mckelvey1995quantal}, QRE has been empirically successful at explaining  deviations from Nash equilibrium predictions across a wide range of experiments \citep*{goeree2016quantal,goeree2020stochastic}. The axiomatic approach has been applied to define non-parametric subclasses of QRE such as regular QRE \citep*{goeree2005regular} and (a)symmetric QRE \citep* {friedman2022quantal} by imposing axioms directly on the quantal response functions. These subclasses impose falsifiable restrictions on behavior, unlike general QRE criticized by \cite*{haile2008empirical}. By contrast, we do not take QRE as a starting point, and instead axiomatize solution concepts rather than response functions.  Furthermore, our results differ in that they pin down the one-parameter class of logit QRE in particular, providing a novel justification for a quantal response function that has been widely used to analyze empirical data \citep*[see, e.g.,][]{goeree2016quantal,wright2017predicting,goeree2020stochastic}.

Our second set of results gives rise to a novel solution concept---statistic response equilibria---where players respond to a {monotone additive statistic} of payoffs, a class of statistics characterized by \cite*{mu2024monotone}.
These equilibria
capture various attitudes toward the uncertainty induced by other players' mixed strategies, including the ambiguity-averse multiplier preferences  \citep*{hansen2001robust,strzalecki2011axiomatic}.
Importantly, these attitudes emerge from our axioms rather than being assumed a priori.
This contrasts with the literature on equilibrium concepts that incorporate specified risk attitudes by transforming game payoffs according to some utility function reflecting such an attitude \citep*{goeree2003risk, yekkehkhany2020risk}. In fact, the risk attitudes resulting from our characterization cannot be replicated by any such transformation of payoffs. In this way, we contribute to the literature on games with non-expected-utility preferences \citep*{shalev2000loss,metzger2019non}. An equilibrium concept related to statistic response equilibria has appeared in the computer science literature, where it has been shown to be efficiently computable under certain conditions \citep*{mazumdar2024tractable}.

\section{Solution Concepts and Axioms}\label{sec_model}

We consider finite normal form games played between a fixed set of $n\geq 2$ players~$\{1,\ldots,n\}$. Such a game $G=(A,u)$ is given by its finite set of action profiles $ A=\prod_i A_i$ and its payoff function $u\colon A \to \R^n$, where $A_i$ is the set of actions of player~$i$ and $u_i(a)$ is the payoff to player~$i$ when an action profile $a$ is played. 

A mixed strategy $p_i \in \Delta A_i$ of player~$i$ is a probability distribution over $A_i$. A mixed strategy profile $p=(p_1,\ldots, p_n)$ induces a product distribution over $A$, also denoted by $p$, with its marginal on $A_{-i}=\prod_{j\ne i} A_j$ denoted $p_{-i}$.  Given $p$, we write $u_i(a_i,p_{-i})$ for the lottery faced by player~$i$ when playing action $a_i$, while other players use mixed strategies~$p_{-i}$. The corresponding expected payoff is denoted by $\mathbb{E}[u_i(a_i, p_{-i})]$.

\smallskip

A \textit{solution concept} $S$ assigns to each game $G$ a nonempty set $S(G) \subset \prod_i \Delta A_i$ of mixed strategy profiles, referred to as \textit{solutions}.\footnote{There is a technical nuance that can be safely ignored without missing the gist of the paper: the collection of all finite sets is not a set, and neither is the collection of all finite games. Hence, for a solution concept to be a well-defined correspondence, we assume that all actions available to any player in any game belong to a universal, non-empty set of actions $\cA$. 
We also suppose that $\cA$ is closed under pairing, so that~$\cA\times \cA \subseteq \cA$.} 
One can think of these solutions as potential predictions for player behavior in $G$. We require that every game be assigned at least one solution. For instance, the map $S(G)=\prod_i \Delta A_i$ assigning to each~$G$ the set of all mixed strategy profiles is a valid---albeit uninformative---solution concept. A more important example of a solution concept is $\Nash(G)$, the correspondence that assigns to $G$ the set of (mixed, including pure) Nash equilibria. Given $\lambda \geq 0$ and a game $G=(A,u)$, the logit quantal response equilibrium correspondence is given by\,\footnote{In a general QRE $p$, each $p_i$ is obtained by applying to the vector of expected payoffs $(\E[u_i(a_i,p_{-i})])_{a_i\in A_i}$ a quantal response function: a continuous map assigning to each payoff vector a mixed strategy \citep*{mckelvey1995quantal}. A QRE is \emph{regular} if this map satisfies monotonicity (actions yielding higher expected payoffs are played with higher probabilities), interiority (full support), and responsiveness (the probability of an action strictly increases in its own expected payoff) \citep*{goeree2005regular}.}
\[
\LQRE_{\lambda}(G)=\Big\{p \in \Pi_i \Delta A_i \mid p_i(a_i) \propto \exp\big(\lambda\E [u_i(a_i, p_{-i})]\big) \; \forall i\Big\}. 
\]
We use $p_i(a_i) \propto \exp(\lambda\E [u_i(a_i, p_{-i})])$ to indicate equality up to normalization of the probabilities.

To compare solution concepts by how specific they are, we say that a solution concept $S$ is a \textit{refinement} of $S'$ if $S(G) \subseteq S'(G)$ for all games~$G$. For example, the solution concept assigning to each game its set of trembling hand perfect equilibria is a refinement of $\Nash$.

\bigskip

Our first axiom is a rationality assumption reflecting that players tend to choose actions with higher payoffs more often.
Given a distribution over actions $p$,
we say that an action $a_i \in A_i$ \emph{first-order dominates} $b_i \in A_i$ if  
the lottery $u_i(a_i,p_{-i})$ first-order stochastically dominates $u_i(b_i,p_{-i})$, which we denote by $u_i(a_i,p_{-i})\geq_{\text{FOSD}}u_i(b_i,p_{-i})$. If moreover the distributions of $u_i(a_i,p_{-i})$ and $ u_i(b_i,p_{-i})$ are not identical, we say that $a_i$ \emph{strictly first-order dominates} $b_i$ and write $u_i(a_i,p_{-i}) >_{\text{{FOSD}}} u_i(b_i,p_{-i})$.\footnote{First-order dominance means that for any threshold $t$, the probability that $u_i(a_i,p_{-i})$ yields a payoff greater than $t$ is at least as high as the probability that $u_i(b_i,p_{-i})$ yields a payoff greater than $t$. Equivalently,   $\E\big[f\big(u_i(a)\big)\big| a_i\big]\geq \E\big[f\big(u_i(a)\big)\big|b_i\big]$ for any nondecreasing payoff transformation $f\colon \R\to\R$. Strict dominance corresponds to strict inequality for any strictly increasing~$f$.}

\begin{definition}
  A solution concept $S$ satisfies \emph{distribution-monotonicity} if $$u_i(a_i,p_{-i}) >_{\text{\emph{FOSD}}} u_i(b_i,p_{-i})\qquad \text{implies}\qquad p_i(a_i) \geq p_i(b_i)$$ for every game $G=(A,u)$, solution $p \in S(G)$, player~$i$, and  actions $a_i,b_i \in A_i$.
\end{definition}
In other words, distribution-monotonicity means that if one action strictly first-order dominates the other, the dominated action cannot be played with a higher probability. Both $\Nash$ and $\LQRE$ satisfy distribution-monotonicity, as do their refinements. Moreover, every regular QRE satisfies distribution-monotonicity.\footnote{\label{footnote:weak-weak}One could imagine a weak-weak version of distribution-monotonicity in which the weak inequality $u_i(a_i,p_{-i})\geq_{\mathrm{FOSD}}u_i(b_i,p_{-i})$ already implies $p_i(a_i)\geq p_i(b_i)$. It has the additional implication that players must break ties uniformly and, hence, is not satisfied by the Nash solution concept, but is satisfied by $\LQRE$. Likewise, a strong-strong version carries the implication that all actions, even strictly dominated ones, must be played with positive probability if they do not yield an FOSD-minimal distribution. Again, the strong-strong version is satisfied by $\LQRE$ but not $\Nash$.}

Distribution-monotonicity is an ordinal axiom that is invariant to monotone transformations of payoffs. Consequently, Nash equilibrium with monotone-reparameterized payoffs \citep*{weinstein2016effect} and risk-adjusted QRE under reparameterized payoffs \citep*{goeree2003risk} also satisfy distribution-monotonicity. Distribution-monotonicity is also satisfied by $S(K)$ equilibria \citep*{osborne1998games}, where players respond to independent draws from the payoff distribution induced by each action. 

This axiom includes two conceptual assumptions. First, it introduces a notion of rational expectations into a solution concept, in the sense that players anticipate the others' actions. Second, it captures a notion of rationality, because players prefer actions with higher payoffs. This is closer in spirit to equilibrium analysis than to rationalizability, which will be reflected in the solution concepts we characterize. Note, however, that this axiom does not entail the assumption that players can perfectly anticipate others' actions. Instead, it only requires that players do this well enough to ``get it right more often than not'' in the sense that they choose a dominating action with higher probability.

\bigskip

Our second axiom,  bracketing, asserts that players who are engaged in multiple unrelated games may consider each game independently. Formally, for games $G=(A,u)$ and $ H=(B,v)$, we define the \emph{composite game} $G \otimes H=(C,w)$ by
\begin{align*}
 C_i=
 A_i \times B_i
 ~~\text{ and }~~   w_i\left((a_{1},b_{1}),\ldots,(a_{n},b_{n})\right)=u_i(a)+v_i(b).
\end{align*}
In $G \otimes H$, each player  $i$ chooses an action $a_i$ from  $A_i$ and an action $b_i$ from $B_i$---effectively playing both games simultaneously---and earns the sum of payoffs from the two games. We call $G$ and $H$ \emph{component games} of the composite game $G \otimes H$.

To make sense of the summation of payoffs across different games, we of course need to think of payoffs as being quoted in the same units, across all games. If payoffs are money, the additive payoff structure of the composite game captures the unrelated nature of the component games, i.e., the action in one does not affect the payoff in the other. If payoffs are utilities, summing them calls for an additional explanation: ordinal utilities cannot be meaningfully added, and von Neumann--Morgenstern utilities capture preferences over lotteries, which are not a part of our approach. We instead adopt a different interpretation of utilities, which we develop in detail in Appendix~\ref{app_generalized_product_games}. The starting point is players' ordinal preferences over physical outcomes of strategic interactions. We derive an \emph{additive cardinalization}: a unique (up to scaling) cardinal representation of a preference that is additive over combined outcomes from interactions that are separable with respect to the underlying preference. In what follows, utilities refer to this additive cardinalization. Accordingly, the additive structure of composite games reflects that players perceive the interactions as separable.

\medskip 

Given mixed strategy profiles $p$ and $q$ for games $G$ and $H$, we define the mixed strategy profile $p \times q$ for the game $G \otimes H$ by 
\begin{align*}
    [p \times q]_i(a_i,b_i) = p_i(a_i) \cdot q_i(b_i).
\end{align*}
So, when players are playing $p \times q$ in $G\otimes H$, they independently choose strategies in $G$ from $p$ and in $H$ from $q$.

\begin{definition}
A solution concept $S$ permits \emph{bracketing} if, for all games $G$ and $H$,
$$p \in S(G)\quad\text{and}\quad q \in S(H)\qquad \text{implies}\qquad p \times q \ \in \  S(G \otimes H).$$
\end{definition}
When a solution concept permits bracketing, solutions of $G$ and $H$ can be combined into a solution of the composite game $G \otimes H$ by having players choose their actions independently in the two component games. Viewing solutions as predictions, our bracketing assumption requires that independently following the predictions in the component games constitutes a valid prediction for the composite game. Importantly, this assumption does not rule out the existence of other predictions for the composite game. Accordingly, rather than restricting the set of predictions that a solution concept gives for a composite game, the bracketing assumption ensures that the set of predictions is rich enough to contain independent predictions.\footnote{Indeed, starting from any solution concept, one can define its closure under bracketing by repeatedly adding all the products to the set of solutions of composite games.}
Notably, bracketing applies only to the very small class of composite games, and does not imply anything for a generic game, which will not be a composite game with exactly additive payoffs, even if the action sets do have a product structure.

The Nash correspondence permits bracketing. Note that not all Nash equilibria in composite games are products of equilibria in the component games, but these products do appear in the solution of the composite game, as required by  bracketing.

Many refinements of Nash also permit bracketing. 
These include maximal-entropy Nash equilibria,\footnote{The maximal-entropy equilibria are the Nash equilibria that maximize the Shannon entropy of the joint distribution of actions.} trembling hand perfect equilibria, and welfare-maximizing equilibria.  
However, not all refinements do. For example, minimal-entropy Nash equilibria---a natural extension of pure Nash equilibria to a non-empty correspondence---violate bracketing as entropy can be reduced by correlating unrelated choices.\footnote{Minimal-entropy Nash equilibria are defined similarly to maximal-entropy ones. To see that this solution concept does not permit bracketing, consider a composition of two copies of matching pennies. The minimal-entropy Nash assigns to matching pennies the unique Nash, and assigns to the composite game the mixed-strategy profiles in which both players perfectly correlate (or anti-correlate) their strategies across the two components. Thus, the product strategy does not appear in the solution of the composite game.} Likewise, a refinement obtained by eliminating Nash equilibria that are Pareto dominated by other Nash equilibria does not
permit  bracketing.\footnote{A standard intuition applies: Pareto optimal allocation in sub-markets may not give rise to a Pareto optimal allocation in the market itself due to beneficial trades across sub-markets.} This is a weakness of our approach, as Pareto optimality is a desirable criterion in many normative settings.

Alongside $\Nash$ equilibria, $\LQRE_\lambda$ also permits  bracketing.
In fact, this correspondence satisfies a stronger property: the set of solutions of a composite game is exactly the set of all products of solutions of the component games.

Other solution concepts that permit  bracketing include the rationalizable mixed strategies, welfare maximizing (or minimizing) mixed strategy profiles, level-$k$ models in which level-$0$ players choose uniformly, as well as cognitive hierarchy models with the same base choices. Probit QRE does not permit bracketing, and (as will follow from our results) neither does any regular QRE that is not logit. The set of all Pareto optimal mixed strategy profiles likewise does not permit bracketing.

We do not think of bracketing as a behavioral axiom, which assumes that players always act independently in composite games. Indeed, such an assumption would require every solution of a composite game to be a product. We think of this axiom as a desirable property of a model, from the point of view of an analyst who has to choose a solution concept. Bracketing is assumed explicitly or implicitly in many models. Allowing for independent behavior enables an analyst to exclude from the model the universe of unrelated factors and considerations, and write a parsimonious description of the situation of interest.

\bigskip

Finally, we also consider anonymity, a simplifying assumption which does not affect the essence of our results. Anonymity requires that permuting players' names results in the corresponding permutation of the solution concept's predictions---that is, the solution concept treats all the players in the same way.
Given a game $G=(A,u)$ and a permutation of players $\pi \colon \{1,\ldots,n\}\to \{1,\ldots,n\}$, define the permuted game $G_\pi = (B,v)$ by $B_i=A_{\pi(i)}$, and $v_i(a_{\pi(1)},\dots,a_{\pi(n)})=u_{\pi(i)}(a_1,\dots,a_{n})$ for all players $i$ and action profiles $a$. Each mixed strategy profile $p$ in $G$ yields a permuted profile $p_\pi$ in  $G_\pi$ by $(p_{\pi})_i=p_{\pi(i)}$.

\begin{definition}
A solution concept $S$ satisfies \emph{anonymity} if for any permutation $\pi$ of players, game~$G$, and its solution $p \in S(G)$, we have $p_\pi \in S(G_\pi)$.
\end{definition}

Our anonymity axiom is a symmetry assumption requiring that players' names do not matter.\footnote{Note that anonymity does not require that every solution to a symmetric game is symmetric, but instead that the set of solutions is symmetric. For example, if $G$ is symmetric then anonymity  means that if $(p_1,p_2) \in S(G)$, then so is $(p_2,p_1)$. 
} In contrast, we nowhere require that action names do not matter. In particular, our setting and axioms do not rule out framing effects, where players behave differently depending on how actions are labeled.

\section{Expectation-Monotonicity and Logit Quantal Response Equilibrium}\label{sec_3}

To develop our intuition for the interaction between bracketing and monotonicity, in this section we strengthen the distribution-monotonicity axiom, and study \emph{expectation-monotonicity}.

\begin{definition}
 A solution concept $S$ satisfies \emph{expectation-monotonicity} if 
 $$\E [u_i(a_i,p_{-i})]>\E [u_i(b_i,p_{-i})]\qquad \text{implies}\qquad p_i(a_i)\geq p_i(b_i)$$ for every game  $G=(A,u)$, solution $p\in S(G)$, player~$i$, and actions $a_i,b_i \in A_i$.
\end{definition} 
That is, actions with strictly lower expected payoffs cannot be played with a higher probability.
Conceptually, this axiom adds an expected utility assumption on top of the rational-expectations and rationality assumptions implied by distribution-monotonicity. Alternatively, if we think of payoffs as money rather than utilities, expectation-monotonicity reflects the assumption that players are risk-neutral. As with distribution-monotonicity, this axiom does not imply that players can calculate expected utilities exactly, but only that they do this well enough, in the sense that they choose actions with higher expected utility more often.

Expectation-monotonicity is satisfied by $\Nash$, $\LQRE_\lambda$, probit QRE, and, more generally, any regular QRE \citep*{goeree2005regular}, and M-equilibrium \citep*{goeree2021m}.\footnote{In fact, \cite{goeree2021m} show that, within any single game, the testable content of regular QRE reduces to a version of expectation-monotonicity.
} It is closed under refinements. The level-$k$ and cognitive hierarchy models do not satisfy this axiom, and neither does the set of rationalizable mixed strategies. Each of the three properties---bracketing, anonymity, and expectation-monotonicity---is satisfied by many well-known solution concepts.  Among the examples mentioned above, $\Nash$ and $\LQRE_\lambda$ satisfy all three, as do trembling hand Nash and maximum entropy Nash. The following theorem shows that any solution concept satisfying all of the three properties must include only Nash equilibria or only logit quantal response equilibria.
\begin{theorem}\label{th_Nash QRE}
    If $S$ permits bracketing and satisfies expectation-monotonicity and anonymity, then~$S$ is a refinement of $\Nash$, or of $\LQRE_\lambda$ for some $\lambda\geq 0$. 
\end{theorem}
Theorem \ref{th_Nash QRE} gives yet another piece of evidence for the importance of Nash equilibria. It also provides a novel justification for the particular logit form of QRE, beyond its tractability.
Furthermore, Theorem~\ref{th_Nash QRE} establishes a connection between Nash and $\LQRE_\lambda$. Notice that $\Nash$ is not merely a limiting case of $\LQRE_\lambda$ as $\lambda\to \infty$; \cite*{mckelvey1995quantal} show that limit points of $\LQRE_\lambda$ are Nash equilibria, but not every Nash equilibrium can be obtained as a limit point of $\LQRE_\lambda$.
Indeed, there is an interesting distinction between logit quantal response equilibria (and their limit points) and Nash equilibria. While all logit quantal response equilibria of a composite game are products of equilibria of its component games, there exist Nash equilibria of composite games that do not satisfy this property. In fact, Nash equilibria exhibit a very rich correlation structure: 
any strategy profile of a composite game that marginalizes to Nash equilibria of its component games constitutes a Nash equilibrium.\footnote{By a result of \cite{solan2020logit}, the whole set of Nash equilibria in a game $(A,u)$ can be obtained as limit points of $\LQRE_\lambda$ as $\lambda \to \infty$, augmented with a vanishing perturbation of the payoff function~$u$. Accordingly, one may view $\Nash$ itself (rather than the limit points of $\LQRE_\lambda$ for fixed $u$) as the natural (yet uncommon) definition of $\LQRE_\infty$. Under this convention, the conclusion of Theorem~\ref{th_Nash QRE} reads simply: $S$ is a refinement of $\LQRE_\lambda$ for some $\lambda \in [0,\infty]$.}

Another property common to $\Nash$ and $\LQRE_\lambda$ is that they assign the same set of solutions to strategically equivalent games. These are games with the same action sets and the same marginal utilities of switching from one action to the other. We explore this notion in Appendix~\ref{app_strategic_eq} and demonstrate that the requirement of expectation-monotonicity in Theorem~\ref{th_Nash QRE} can be relaxed to distribution-monotonicity, together with the requirement that strategic equivalence is respected by the solution concept. This result further highlights the importance of these solution concepts and their connection.

\bigskip

The theorem is proved in Appendix~\ref{ap_A}. We illustrate some of the ideas here. Suppose $S$ satisfies expectation-monotonicity,  bracketing, and anonymity. 
Consider a class of games, where all players but player~$1$ are dummies facing no strategic choice, and player~$1$ chooses between two actions, one of which is dominant. In such a game $G_x$, indexed by $x>0$, the action sets are $A_1=\{h,\ell\}$ and $A_i=\{c\}$ for~$i\neq 1$, and
the payoff for the first player is $x$ when playing $h$ and $0$ when playing $\ell$. We will refer to games $G_x$ as \emph{test games} as the behavior of $S$ on $G_x$ turns out to determine its behavior elsewhere.

Consider a solution $p \in S(G_1)$.
We consider two cases depending on whether the dominated action is excluded ($p_1(\ell)=0$) or not ($p_1(\ell)>0$). In the former case, it turns out that $S$ is a refinement of $\Nash$, and in the latter, a refinement of $\LQRE$.
The case of $p_1(\ell) = 0$ is treated in the following claim:
\begin{claim}
\label{cl_dominated}
If $p_1(\ell)=0$, then $S$ is a refinement of Nash equilibrium.
\end{claim}

\begin{proof}
The solution $p\in S(G_1)$ with $p_1(\ell)=0$ coincides with the $\Nash$ prediction for this test game. We now show that for any game $H=(B,v)$, each $q\in S(H)$ is a Nash equilibrium as well. For the sake of contradiction, suppose $\E[v_i(b,q_{-i})]>\E[v_i(a,q_{-i})]$ while $q_i(a) > 0$ for some player~$i$ and some actions $a,b \in B_i$. By anonymity, we can assume that $i=1$. Choose \[n>\frac{1}{\E[v_1(b,q_{-1})]-\E[v_1(a,q_{-1})]},\] and consider the composite game $H^n\otimes G_1$, where $H^n$ denotes the $n$-fold composition $H \otimes \cdots \otimes H$.
By  bracketing, $q^n \times p$ (where  $q^n=
 q \times \cdots \times q$) is a solution for this composite game. In the composite game, the expected payoff to player~1 for playing $b$ in every copy of $H$ and playing $\ell$ in $G_1$ is strictly higher than playing $a$ in every copy of $H$ and playing $h$ in $G$. This violates expectation-monotonicity since $q^n \times p$ places positive probability on $(a,\ldots,a,h)$ and zero probability on $(b,\ldots,b,\ell)$, since $p_1(\ell)=0$. Hence players never play strategies that yield lower expected payoffs, i.e., they play a Nash equilibrium.
\end{proof}

We now discuss the case $p_1(\ell) 
>0$, which corresponds to $\LQRE_\lambda$. Unlike $\Nash$, this rule is sensitive to magnitudes of payoff differences. As we explain below, the exponential function capturing this sensitivity originates from an equation linking $S(G_x)$ across test games $G_x$, while the constant $\lambda$ is determined by the solution $p$ to any single test game $G_x$.

As the first step toward deriving an equation linking behavior across games, we show in the appendix that $p_1(\ell) >0$ implies two general properties of $S$: interiority and expectation-neutrality. A solution concept $S$ satisfies \emph{interiority} if every action in every game is played with positive probability, that is, for every game $G=(A,u)$, solution $p \in S(G)$, and player~$i$, we have $p_i(a_i)>0$ for each $a_i \in A_i$. It satisfies \emph{expectation-neutrality} if actions yielding equal expected payoffs are played with equal probabilities:  $\E [u_i(a_i,p_{-i})] = \E [u_i(b_i,p_{-i})]$ implies $p_i(a_i)= p_i(b_i)$.

Next, we consider three test games $G_x,G_y$ and $G_{x+y}$ for arbitrary payoffs $x$ and $y$ and derive an equation relating solutions assigned by~$S$ to these games. This step is an adaptation to our setting of the proof strategy used by \cite*{sandomirskiy2025iid} for single-valued predictions in a single-agent setting. We illustrate the core argument further assuming that the solution $S$ is single-valued, and let $r_x$ denote the probability with which player~$1$ chooses the dominant action  $h$ according to $S(G_x)$.
In the composite game $G_x \otimes G_y \otimes G_{x+y}$, the action tuples $(h, h, \ell)$ and $(\ell, \ell, h)$ both give player~$1$ a total payoff of $x + y$. By expectation-neutrality, these two profiles must be chosen with equal probability, leading to the identity
$$
r_x r_y \big(1 - r_{x+y}\big) = (1 - r_x)(1 - r_y)r_{x + y}.
$$
By interiority, none of the factors is zero, so we can take logarithms on both sides. Defining 
$
f(x) = \log\big({r_x}/{(1 - r_x)}\big),
$
we obtain the Cauchy functional equation $f(x) + f(y) = f(x + y)$. The function $f$ must be monotone to be compatible with expectation-monotonicity. Since the only monotone solutions to the Cauchy equation are linear functions, \( f(x) = \lambda\cdot  x \), where $\lambda \geq 0$ is uniquely determined by the probability of $h$ in any game  $G_x$. We conclude that $S$ is a refinement of $\LQRE_\lambda$ on the family of test games $G_x$. The details of this argument and its extension to arbitrary games $G$ are contained in the appendix.

\medskip
 We note that anonymity is not a crucial assumption for Theorem~\ref{th_Nash QRE}, but rather ensures that all players behave alike. Without anonymity, we get a version of $\LQRE_\lambda$ with player-specific $\lambda_i$ instead of a common $\lambda$, as well as chimeric rules where some agents use logit best responses and others best-respond as in Nash equilibrium.\footnote{\cite*{mckelvey2000effects} extended the QRE framework to allow for such $\lambda$-heterogeneity. Interestingly, they experimentally reject the hypothesis that $\lambda$ is fixed across games, which is a consequence of our axioms.} As a corollary, we do not require anonymity to characterize refinements of Nash equilibrium (even though it is an anonymous solution concept).

\begin{corollary}\label{cor_only Nash}
   If $S$ permits bracketing, satisfies expectation-monotonicity, and players never play strictly dominated actions, then $S$ is a refinement of $\Nash$. 
\end{corollary}

The assumption on strictly dominated actions serves to distinguish Nash from $\LQRE$. In fact, under the assumptions of expectation-monotonicity and  bracketing, \textit{any} feature of Nash equilibrium that does not apply to $\LQRE$ will lead to a characterization of only Nash equilibrium and vice versa. For example, we can weaken the assumption on strictly dominated actions by supposing that every player plays some action with zero probability in some game. Similarly, one can characterize $\LQRE$ by augmenting the axioms of Theorem~\ref{th_Nash QRE} with, for example, the assumption of interiority or, alternatively, with expectation-neutrality. 

In Appendix~\ref{sec_bnb}, we compare this characterization of Nash equilibrium with that of \cite*{brandl2024axiomatic}. 
Their axioms rule out non-trivial refinements of Nash equilibrium, 
whereas our axioms allow for some refinements, such as trembling hand perfect equilibrium. Recall that not all refinements of Nash permit bracketing, which provides a justification for selecting some refinements of Nash over others.

It is also instructive to compare our bracketing axiom with the classical IIA axiom used by \cite*{luce1959individual} to characterize multinomial logit, i.e., LQRE for single-agent decisions. Luce's axiom stipulates that the ratio between the choice probabilities of two given alternatives is the same in any two menus containing them. Bracketing relates a composite game to its component games, and does not impose any restrictions on behavior in pairs of games where neither is a component of the other.
Moreover, bracketing by itself does not pin down LQRE, and is satisfied by many solution concepts, such as level-$k$ reasoning with uniform mixing at level $0$.

One could imagine an alternative approach to axiomatizing LQRE based on an appropriate extension of the IIA axiom. However, there is a conceptual difficulty, as highlighted by the following example.
In Table~\ref{tab:IIA}, we consider two games that differ by $c_1$, which is an additional action available to player~1 in the game on the right. 
     \begin{table}[h]
    \centering
    \begin{tabular}{c|c|c}
            & $a_2$&$b_2$ \\ \hline
         $a_1$&  $(0,0)$& $(2,0)$\\
         $b_1$&  $(1,0)$& $(1,0)$\\
    \end{tabular}
\qquad
\qquad
    \begin{tabular}{c|c|c}
            & $a_2$&$b_2$ \\ \hline
         $a_1$&  $(0,0)$& $(2,0)$\\
         $b_1$&  $(1,0)$& $(1,0)$\\
         $c_1$&  $(\alpha,\beta)$& $(\gamma,\delta)$\\
    \end{tabular}
    \caption{Additional actions are not irrelevant}
    \label{tab:IIA}
\end{table}

In any $\LQRE$, the unique solution to the left game is uniform mixing by both players. A naive adaptation of the IIA property to games would require any solution $q$ to the game on the right to also satisfy $q_1(a_1)=q_1(b_1)$, since the two actions have the same payoff profile as they do in the game on the left. However, unless $\beta=\delta$, player~2 will not randomize uniformly under $q$ (provided $\lambda \neq 0$). Thus $q_1(a_1) \neq q_1(b_1)$ for all $\LQRE_\lambda$, with~$\lambda \neq 0$.

\section{Distribution-Monotonicity and Statistic Response Equilibria}\label{sec:SRE}

In this section we show that the only solution concepts permitting bracketing and satisfying anonymity and distribution-monotonicity are \textit{statistic response equilibria}, where players respond to a statistic of each action's payoff distribution. This class of equilibria generalizes $\Nash$ and $\LQRE_\lambda$---in which players evaluate actions by their expected payoffs---and accommodates other preferences over payoff distributions (which we refer to as lotteries). Notably, our characterization implies that players behave as if driven by a single preference for lotteries that is stable across all games.

We use the term \textit{statistic} to refer to a function $\Phi$ that assigns a real number to every lottery  and satisfies $\Phi[c]=c$ for deterministic lotteries yielding a constant amount $c$. Here, a lottery is simply a finitely supported distribution over the real line. Lotteries arise in our setting as the payoffs a player anticipates when choosing an action, given the mixed strategies of the other players. We denote lotteries by $X,Y$, and $X+Y$ denotes the lottery corresponding to the sum of outcomes drawn independently from $X$ and $Y$; that is, the distribution of $X+Y$ is the convolution of the distributions of $X$ and~$Y$.

We now define a class of statistics that are monotone with respect to first-order stochastic dominance and additive for independent lotteries.
Below, we demonstrate that players must respond to this class of statistics.

\begin{definition}
    A statistic $\Phi$ is a \emph{monotone additive statistic} if 
    $$\Phi[X+Y]=\Phi[X]+\Phi[Y]\qquad \text{and}\qquad \Phi[Z]\leq \Phi[W]\quad \text{for}\quad Z\leq_{\text{\emph{FOSD}}} W.$$
\end{definition}
A canonical example of a monotone additive statistic is given by the normalized cumulant-generating function of~$X$,
defined as $K_a[X]=\frac{1}{a}\log \E\left[\ee^{aX}\right]$ for $a \in \R$. Taking limits as $a$ approaches~$\pm\infty$ and~$0$, we obtain $K_{-\infty}[X]$, $K_0[X]$, and $K_\infty[X]$, which correspond to the minimum, expectation, and maximum of $X$, respectively.
When $X$ is a random amount of money, the statistic $K_a[X]$
has an important economic interpretation: each $K_a[X]$ represents the certainty equivalent of lottery $X$ for an agent with constant absolute risk aversion (CARA) utility, where $-a$ is the coefficient of absolute risk aversion; negative $a$ values correspond to risk aversion, and positive values to risk-loving preferences. When $X$ is a random utility, $K_a[X]$ represents a multiplier preference, which can be interpreted as a form of ambiguity aversion \citep*{hansen2001robust, strzalecki2011axiomatic}. \cite*{mu2024monotone} show that this example is, in fact, universal: any monotone additive statistic can be represented as $\Phi[X]=\int_{\overline{\R}}K_a[X]\,\dd \mu(a)$ for some probability measure $\mu$ on the extended real line $\overline{\R}=\R\cup\{+\infty,-\infty\}$. Hence, any monotone additive statistic is a weighted average of statistics of the form $K_a$.\footnote{\cite*{mu2024monotone} provide a characterization of monotone additive statistics on the domain of all compactly supported (rather than finitely supported) lotteries. Their characterization also applies to the domain of finitely supported lotteries, since any monotone additive statistic on this restricted domain can be extended to the compactly supported lotteries. For a proof, see Appendix \ref{ap_MAS}.}

In Nash equilibrium and in $\LQRE_\lambda$ players respond to the expectation of each action's payoff distribution. In a statistic response equilibrium, players evaluate actions using a monotone additive statistic of the payoff distribution.
Below we define two classes of statistic response equilibria.

\begin{definition}
Given a monotone additive statistic $\Phi$, a mixed strategy profile $p$ is a $\Nash_\Phi$ \emph{equilibrium} of the game $G=(A,u)$ if for all players~$i$
\[\supp p_i\subseteq \argmax_{a_i} \Phi\big[u_i(a_i,p_{-i})\big].\]
\end{definition}

In a $\Nash_\Phi$ equilibrium, players best respond to the other players' mixed strategies according to $\Phi$ by randomizing over actions whose payoff distributions maximize $\Phi$.  
The next definition introduces a class of statistic response equilibria in which players ``better respond'' to a monotone additive statistic of each distribution. 
\begin{definition}
Given a monotone additive statistic $\Phi$ and $\lambda\geq 0$, a mixed strategy profile $p$ is an $\LQRE_{\lambda\Phi}$ \emph{equilibrium} of the game $G=(A,u)$ if for all players~$i$ and actions $a_i \in A_i$ \[p_i(a_i) \propto \exp\big(\lambda\Phi \big[u_i(a_i,p_{-i})\big]\big).\]
\end{definition}

$\Nash_\Phi$ and $\LQRE_{\lambda\Phi}$ generalize  Nash and logit quantal response equilibria, in which players respond to the expectation, i.e.,  $\Phi=\E$.\footnote{Despite a superficial resemblance of the algebraic expression, $\LQRE_{\lambda\Phi}$ does not reduce to random coefficients logit for single-agent decision problems; in mixed logit choice probabilities are calculated by integrating the exponent, rather than exponentiating the integral as in $\LQRE_{\lambda\Phi}$.} In general, the preference expressed by a monotone additive statistic $\Phi$ is not equivalent to an expected utility preference. As we show in Appendix~\ref{app_CARA}, the only exceptions are $\Phi = K_a$ for $a \in \R$. Since these correspond to multiplier preferences,  $\Nash_\Phi$ and $\LQRE_{\lambda\Phi}$  incorporate in games some important preferences over gambles that have been previously explored in the decision theory literature.

While every game has a $\Nash$ equilibrium, the existence of a $\Nash_\Phi$ equilibrium is not guaranteed for all $\Phi$. 
For example, $\Nash_\Phi$ equilibria may not exist when $\Phi$ is the minimum or maximum of a distribution. This issue does not arise for $\LQRE_{\lambda\Phi}$, which do exist for every game. As the next result shows, the existence of $\Nash_\Phi$ equilibria is also guaranteed for a large family of monotone additive statistics, namely those in which the maximum and minimum do not play a role.

\begin{proposition}\label{pr_Nash_Phi_Exist}
    Let $\Phi=\int_{\overline{\R}} K_a\, \dd \mu(a)$ be a monotone additive statistic. Then
    \begin{itemize}
    \item There is an $\LQRE_{\lambda\Phi}$ equilibrium for every game and $\lambda \geq 0$.
    \item There is a $\Nash_\Phi$ equilibrium for every game if and only if  $\mu(\{-\infty\})=\mu(\{+\infty\})=0$.
    \end{itemize}
\end{proposition}
Proposition \ref{pr_Nash_Phi_Exist} is proved in Appendix \ref{ap_pr_1}. The existence of an $\LQRE_{\lambda\Phi}$ equilibrium follows from Brouwer's fixed-point theorem, applied to a version of the logit response function; the latter must be appropriately modified to ensure continuity  when $\mu$ places positive weight on $\pm\infty$.
The proof that $\Nash_\Phi$ equilibria exist when $\mu$ places no mass on the minimum or maximum follows from a standard fixed point argument.  This argument does not apply when $\mu$ places positive mass on the minimum or maximum, since $\Phi$ may be discontinuous at strategy profiles in which some action of an opponent has probability zero: there, the support of the induced payoff lottery changes abruptly, and so can the minimum or maximum. Indeed,  $\Nash_\Phi$  equilibria may fail to exist for such~$\mu$, as we show using a variant of matching pennies in Appendix \ref{ap_pr_1}. 

\medskip

We call $\Nash_\Phi$ and $\LQRE_{\lambda\Phi}$ \textit{statistic response equilibria} as players best or better respond to the statistic $\Phi$ of distributions induced by each available action. Formally, a statistic response equilibrium (SRE) is a solution concept that returns all $\Nash_\Phi$ or all $\LQRE_{\lambda\Phi}$ equilibria for some $\Phi$ and $\lambda$.\footnote{By Proposition~\ref{pr_Nash_Phi_Exist}, $\Nash_\Phi$ is a well-defined solution concept only for $\Phi$ that puts no mass on the maximum or minimum.}

It is easy to verify that the SRE solution concepts satisfy our axioms. Bracketing follows from the additivity of $\Phi$, distribution-monotonicity is a consequence of the monotonicity of~$\Phi$, and anonymity holds since all players use the same $\Phi$.  
The next result shows that these axioms, in fact, characterize SRE.
\begin{theorem}\label{th_sre} 
If $S$ permits bracketing and satisfies distribution-monotonicity and anonymity, then~$S$ is a refinement of some SRE. 
\end{theorem}
Anonymity in Theorem \ref{th_sre} can be removed as in Theorem~\ref{th_Nash QRE}, with the conclusion appropriately altered to allow different players to best or better respond to different statistics. 

While distribution-monotonicity ensures that players' mixing probabilities are monotone with respect to first-order stochastic dominance, it does not provide a way to compare every pair of distributions. Interestingly, 
Theorem \ref{th_sre} shows that, with bracketing, there is a total order, defined by a statistic $\Phi$, that dictates how players rank \textit{every} payoff distribution. Hence, Theorem \ref{th_sre} gives rise to preferences over lotteries, which are not a primitive of our model.

We explore two different interpretations of this result. If payoffs are utilities, the statistic~$\Phi$ that players respond to is now (in general) not the expectation of these utilities. Indeed, it is in general not even the expectation of any monotone transformation of these utilities, as we explain in Appendix~\ref{app_CARA}. Thus, SREs incorporate non-expected-utility behavior into games, yet do so in a principled way that preserves monotonicity and bracketing. This flexibility allows a range of behavior that is larger than what can be accommodated within the expected-utility framework. In \S\ref{sec_min_max_mean} we explore a parametrized family of SREs and explain how they can predict behavior that is compatible with the Allais paradox, for example.

In another interpretation,  payoffs are monetary. The statistic $\Phi$ to which players respond is then a weighted average of CARA certainty equivalents across different coefficients. Hence statistic response equilibria incorporate flexible risk attitudes which allow for risk-averse, risk-loving, or mixed risk attitudes.
Indeed, consider a statistic $\Phi[X] = \int_{\overline \R}K_a[X]\, \dd \mu(a)$. If $\mu$ places mass only on negative values of $a$, $\Phi[X] \leq \E[X]$ for any lottery $X$, i.e., $\Phi$ reflects risk-aversion. Conversely, if $\mu$ places mass only on positive values of $a$, then $\Phi[X] \geq \E[X]$ and $\Phi$ reflects a risk-loving attitude. If $\mu$ places mass on both negative and positive values of $a$, then $\Phi$ reflects a mixed risk attitude, i.e. there are lotteries $X$ and $Y$ with $\Phi[X] < \E[X]$ and $\Phi[Y] > \E[Y]$; see Proposition~5 of \cite*{mu2024monotone}.

\bigskip

The proof of Theorem~\ref{th_sre} begins with the observation (formalized in Proposition~\ref{lm_Ghl} below) that solution concepts satisfying our axioms fall into one of two classes: {ordinal-Nash or ordinal-QRE}.
\begin{definition}
    A mixed strategy profile $p$ for a game $G=(A,u)$ is an \emph{ordinal-Nash} if $p_i(a_i)=0$ whenever there exists $b_i \in A_i$ such that $u_i(a_i,p_{-i}) <_\text{FOSD} u_i(b_i,p_{-i})$.
\end{definition}
That is, $p$ is an ordinal-Nash if players never play an action whose outcome distribution is strictly first-order dominated by that of another action. This is a very permissive solution concept that generalizes Nash. It retains the rational-expectations aspect of Nash equilibria, but completely discards cardinal expectation-maximization, leaving only a weaker ordinal assumption. 
Beyond Nash, every $\Nash_{\Phi}$ 
is also a refinement of ordinal-Nash.  Clearly, ordinal-Nash satisfies distribution-monotonicity and is itself a refinement of the solution concept consisting of all strategy profiles that survive repeated elimination of strictly dominated strategies.

The next class is an analogous weakening of QRE:
\begin{definition}
    A mixed strategy profile $p$ for a game $G=(A,u)$ is an \emph{ordinal-QRE} if $p_i(a_i)>0$ for all $i$ and $a_i \in A_i$ (interiority), and $u_i(a_i,p_{-i}) \geq_{\text{FOSD}} u_i(b_i,p_{-i})$ implies $p_i(a_i) \geq p_i(b_i)$.
\end{definition}
In this solution concept, strategies satisfy a strengthening of distribution-monotonicity, which, beyond  monotonicity, implies that actions yielding equal outcome distributions are played with the same probability; we call this property \emph{distribution-neutrality}. In the spirit of QRE, it also includes an interiority assumption, and indeed every regular QRE is a refinement of ordinal-QRE, as is every $\LP$. As with ordinal-Nash, this is a very permissive concept that includes no cardinal assumptions, but retains distribution-monotonicity.

\begin{proposition}\label{lm_Ghl}
    If $S$ permits bracketing and satisfies distribution-monotonicity and anonymity, then exactly one of the following two statements holds:
    \begin{enumerate}
        \item $S$ is a refinement of ordinal-Nash.
        \item $S$ is a refinement of ordinal-QRE.
    \end{enumerate}
\end{proposition}

Proposition \ref{lm_Ghl} is proved in Appendix~\ref{ap_distr_neutrality}. To convey some of the ideas of the proof, recall the classical notion of dominated actions. An action $a_i$ \emph{strictly dominates}  $b_i$ if $u_i(a_i,a_{-i}) > u_i(b_i,a_{-i})$ for all~$a_{-i}$.
Note that $a_i$ strictly dominates $b_i$ 
if and only if $b_i$ is strictly first-order dominated by $a_i$ for every $p$. However, $a_i$ can strictly first-order dominate $b_i$ under some profile~$p_{-i}$ but not others, in which case there is no strict dominance.

We show that the two cases in the proposition correspond to whether or not players ever play strictly dominated actions. Similarly to the proof of Theorem~\ref{th_Nash QRE}, if such actions are ever played, it can be shown that the solution concept satisfies interiority and distribution-neutrality.
The key difficulty in the proof is demonstrating that by ruling out strictly dominated actions we inevitably rule out all first-order dominated actions as well. This requires the construction of an interesting game (Table~\ref{tab:VMP_low_maxmin}) and a novel result about stochastic dominance (Lemma~\ref{lm_K_trick}).

\medskip

We now explain how we use Proposition~\ref{lm_Ghl} to complete the proof of Theorem~\ref{th_sre}. The two cases in Proposition~\ref{lm_Ghl} correspond to the two families of SREs. Like in the proof of Theorem~\ref{th_Nash QRE}, which led to the characterization of $\Nash$ and $\LQRE$, we construct a family of test games designed so that the behavior of $S$ on this family fully determines its behavior in general.
The SRE solution concepts, which we ultimately aim to characterize, depend on a statistic $\Phi$. Accordingly, we require a sufficiently rich family of test games to identify the statistic. This leads to a more intricate construction than the earlier games $G_x$, where player~1 simply compared deterministic payoffs of $0$ and $x$ and all other players were dummies. In contrast, for Theorem~\ref{th_sre}, player~1 must compare a given lottery $X$ with a deterministic payoff $r$. Crucially, the lottery $X$ must arise endogenously from the strategies of other players, who therefore can no longer all be dummies.

We now describe the construction of test games that guarantee player~1 faces this endogenous choice between $X$ and $r$ in every solution. The construction, based on Proposition~\ref{lm_Ghl}, varies depending on which of its two cases applies.

In the first case, we consider $S$ that is a refinement of ordinal-Nash. The test game is constructed as follows. Let $X$ be a lottery with finite support and rational weights. Any such lottery can be generated by sampling a coordinate of some vector $x=(x_1,\ldots,x_m)$ uniformly at random. Given such a vector $x\in \R^m$, a number $r\in \R$, and small $\varepsilon>0$, we define a game $G_{r,x,\varepsilon}$.
In this game, player~1 selects $a_1\in \{a_x,a_r\}$ \emph{and} a permutation $\pi$ of $\{1,\ldots, m\}$. Player~2 selects a number $a_2\in \{1,\ldots, m\}$. The resulting payoffs  for player~1 are
\begin{align}
    \label{eq:G_rXeps}
    u_1\big((a_x,\pi),\,a_2\big)= x_{\pi(a_2)},\qquad \qquad 
    u_1\big((a_r,\pi),\,a_2\big)=r+\varepsilon x_{\pi(a_2)}.
\end{align}
For player~2, 
\begin{align*}
u_2\big((a_1,\pi),a_2\big) = -x_{\pi(a_2)}.
\end{align*}
This game admits the following interpretation. Player $1$ has $m$ cards, each showing some amount of money $x_i$. Player $1$ chooses a shuffling $\pi$ of the cards. Player $2$ chooses one of these cards, and pays player~$1$ the amount of money the card shows. Depending on whether player~$1$ chose $a_x$ or $a_r$, player~1 either keeps the money paid by player~$2$ or else player~$1$ gains~$r$, but has to burn a $1-\varepsilon$ fraction of the money received from player~$2$. From the point of view of player~1, the choice between $a_r$ and $a_x$ is a choice between the almost-deterministic payoff $r$ (assuming $\varepsilon$ is small) and the random amount of money transferred by player~2. From the point of view of player~2 who transfers money to player~1, this game is very similar to a zero-sum game, since player~2 does not care if player~1 burns some of the transferred money.

To gain intuition about this game, consider its Nash equilibria first. Suppose player~2 does not choose uniformly at random. Then player~1 will choose a shuffling that takes advantage of this, making player~2 more likely to pick a higher card. But knowing this, player~2 would want to pick a different strategy, in order to pay less in expectation. Thus, in equilibrium, player~2 must mix uniformly.
We show in Lemma~\ref{lm_lotteries_games_FOSD_undominated} that this conclusion does not require the assumption that players play a Nash equilibrium: player~2 must mix uniformly in every ordinal-Nash. Note that it is important that $\varepsilon>0$. For $\varepsilon=0$, it is not guaranteed that player~2 will mix uniformly. E.g., if $r>\max_j x_j$, there is a pure Nash equilibrium where player~1 selects $a_r$ and the identity permutation and player~2 chooses the lowest card $\argmin_j x_j$.

The uniform mixing of player~2 makes player~1 choose between the action $a_x$ that yields the lottery $X$ and $a_r$ that yields the (approximately) deterministic payoff $r$. Accordingly, the test games provide a measure of how players evaluate a rich class of lotteries for any $S$ that is a refinement of ordinal-Nash. Indeed, the highest $r$ for which player~1 chooses $a_x$ in $G_{r,x,\varepsilon}$ plays the role of an approximate certainty equivalent of the lottery $X$ yielded by~$a_x$. Formally, we define $\Phi$ as the following limit as $\varepsilon\to 0$:
\begin{align*}
    \Phi[X]=\lim_{\varepsilon\to 0} \left( \sup \Big\{ \ r\in \R \ \ \Big\vert  \    p_1(a_x,\pi )>0   \text{ for some }  \pi  \text{ and }    p \in S\big(G_{r,x,\varepsilon}\big)  \  \Big\}\right).
\end{align*}
The distribution-monotonicity and bracketing axioms imply the existence of the limit and that the recovered $\Phi$ is a monotone additive statistic for lotteries with rational weights. 
We then show that $\Phi$ extends uniquely to general lotteries that can include irrational weights. Finally, we use the axioms to demonstrate that the same statistic $\Phi$ is deployed to compare actions in every game, leading to $S$ being a refinement of $\Nash_\Phi$. The details are provided in the appendix.
\smallskip

In the second case of Proposition~\ref{lm_Ghl}, we consider $S$ that is a refinement of ordinal-QRE. In this setting, it is more straightforward to construct a sufficient class of test games, where player~1 is guaranteed to face an action that yields a given lottery $X$ and an action that yields a deterministic payoff $r$. This is due to distribution-neutrality, by which player~2 must uniformly mix between actions that result in the same deterministic payoffs, allowing us to generate uniform mixing by setting $u_2\equiv 0$. Like the test games constructed for the first case of Proposition~\ref{lm_Ghl}, these test games provide a measure of how players evaluate a rich class of lotteries, giving rise to a monotone-additive statistic $\Phi$. We then show that players logit respond to $\Phi$ in every game. This is demonstrated in a manner similar to the $\LQRE$ case of Theorem~\ref{th_Nash QRE}.

\section{Min-Max-Mean Response Equilibria}\label{sec_min_max_mean}

In statistic response equilibria, players respond to a statistic $\Phi$, which is parameterized by a probability measure $\mu$ on $\overline{\R}$, an infinite-dimensional parameter. Having such a large parameter space can be challenging for empirical studies, especially on small datasets. In this section, we introduce an additional axiom that singles out a tractable parametric family of SREs that is rich enough to be empirically viable. In general, a monotone additive statistic $\Phi$ can be viewed as an average of the statistics $K_a[X] = \frac{1}{a}\log\E[\ee^{aX}]$. For high values of $a$, $K_a$ puts a lot of weight on the maximum of the distribution of $X$, in the sense that $K_a[X]$ tends to the maximum of $X$ as $a$ tends to infinity. For very low negative values of $a$, almost all the weight is on the minimum. For $a=0$, $K_a$ is simply the expectation, which in some sense puts very little weight on the minimum or maximum. 

The parametric family that we consider consists of all the SREs in which $\Phi$ is just an average of the minimum $K_{-\infty}$ (worst-case), maximum $K_{\infty}$ (best-case), and expectation $K_0$
(average-case). This family thus captures much of the qualitative diversity of the family of SREs, while being much smaller, as it is parametrized by just three variables. We call this family Min-Max-Mean response equilibria.
\begin{definition}
Given $\lambda \in \R_{\geq 0}^3$, a mixed strategy profile $p$ is a $\text{Min-Max-Mean}_\lambda$ equilibrium of a game $G=(A,u)$ if 
    \[
      p_i(a_i) \propto \exp\left(\lambda_1 \min_{a_{-i}} u_i(a_i,a_{-i})+\lambda_2\E[u_i(a_i,p_{-i})]+\lambda_3\max_{a_{-i}} u_i(a_i,a_{-i})\right)
    \]    
for all  players~$i$ and actions $a_i \in A_i$.
\end{definition}
In this definition, the minima and maxima are taken over all opponent action profiles $a_{-i}\in A_{-i}$. These coincide with the minima and maxima of $u_i(a_i,p_{-i})$, since every action is played with positive probability.

Min-Max-Mean response equilibria may reflect how players react to complexity in games. Computing the minimum and maximum of each induced lottery is straightforward and does not require any strategic reasoning. Empirically, players take the minimum and maximum into account when they face multiple games simultaneously \citep*{avoyan2020attention}. 
Moreover, assuming that other players are affected by the minimum and maximum of the lotteries they face reduces the cognitive load of hypothesizing about opponents' behavior. There is experimental evidence that players assume that their opponents use simple heuristics in choosing their strategies.
\cite*{friedenberg2024beyond} look at three heuristics: the maximum, the minimum, and the sum of payoffs. One of their findings is that 88\% of subjects play as if they believe that, with some probability, their opponents will use one of these heuristics rather than behave rationally.

Min-Max-Mean evaluation also admits a decision-theoretic interpretation in terms of ambiguity. For a fixed profile $p_{-i}$ of the opponents, $\text{Min-Max-Mean}_\lambda$ equilibrium corresponds to players logit responding to $\alpha$-maxmin expected utility, given by \[\alpha\min_{r_{-i}\in C}\E[u_i(a_i,r_{-i})]+(1-\alpha)\max_{r_{-i} \in C}\E[u_i(a_i,r_{-i})],\] where $C=\{(1-\varepsilon)p_{-i}+\varepsilon q_{-i} \mid q_{-i} \in \Delta(A_{-i})\}$ and $\varepsilon=\frac{\lambda_1+\lambda_3}{\lambda_1+\lambda_2+\lambda_3}$ and $\alpha=\frac{\lambda_1}{\lambda_1+\lambda_3}$.\footnote{{In this interpretation, player $i$ entertains not a single belief about the play of the opponents but the set $C$ of candidate beliefs. Sets of the form $C$ are known as $\varepsilon$-contamination sets of priors. They represent players who somewhat trust the reference belief $p_{-i}$, but believe that the opponents will play arbitrary $q_{-i}$ with probability $\varepsilon$. The parameter $\alpha$ captures players' optimism/pessimism. See \cite*{gajdos2008attitude} for background on these representations. We thank an anonymous referee for suggesting this connection.}}

All SREs are invariant to adding constants to players’ utilities. More generally, they are invariant to adding the same independent lottery to every outcome; this follows directly from the defining property of the monotone additive statistics that underlie SREs. To characterize the class of Min-Max-Mean equilibria, we supplement this invariance with a requirement of scale-invariance.

\begin{definition}\label{def_scale_inv}
    $S$ satisfies \emph{scale-invariance} if whenever $p_i$ is the uniform distribution on $A_i$ for all $i$, $p\in S(A,u)$ implies $p \in S(A,\alpha\cdot u)$ for all $\alpha \in (0,1)$. 
\end{definition}
That is, reducing all payoffs by the same positive factor does not change whether the uniform mixed strategy is a solution.\footnote{This requirement follows from the consistency and consequentialism axioms of \cite*{brandl2024axiomatic}; see Appendix~\ref{sec_bnb}.}
This assumption is weak, along two dimensions. First, it only applies when all players mix uniformly. In a way, this means that it only restricts players' behavior when all players are indifferent among all actions.\footnote{A longer but more cumbersome name such as indifference-scale-invariance might be more appropriate.} Second, the restriction is imposed only
for scales $\alpha$ less than unity. Intuitively, it seems plausible that if a player is indifferent between all actions, then they would still be indifferent when the stakes are made lower.\footnote{In fact, the proof of Theorem~\ref{th_complex} uses the axiom for a single scale, $\alpha=\frac{1}{2}$, and any fixed positive $\alpha\neq 1$ would do. We state the axiom for $\alpha \in (0,1)$ because lowering stakes seems the more natural restriction for maintaining indifference.}

\begin{theorem}\label{th_complex}
    Suppose $S$ permits bracketing and satisfies 
      distribution-monotonicity, anonymity, scale-invariance, and interiority. Then $S$ is a refinement of $\text{Min-Max-Mean}_\lambda$ for some $\lambda \in \R_{\geq 0}^3$.
\end{theorem}
We assume interiority to keep the statement of the theorem simple; removing this assumption would only add back the Nash solution concept and some of its refinements.\footnote{Indeed, by Theorem~\ref{th_sre}, a non-interior $S$ must be a refinement of some $\Nash_\Phi$, and scale-invariance forces $\Phi$ to be a combination of the minimum, maximum, and expectation. Proposition~\ref{pr_Nash_Phi_Exist} shows that only $\Phi=\E$ yields a well-defined solution concept, i.e., Nash equilibrium.}
Perhaps surprisingly, scale-invariance rules out all statistics except the extreme ones and the expectation. While this family is simple enough to be represented with only three parameters, it is rich enough to account for a wide range of behavior beyond expected-utility preferences, such as that exhibited in the Allais paradox.

 \begin{table}[ht]
    \centering
    \begin{tabular}{|c|c|c|c|}
    \hline
         Probability &  $89\%$ & $1\%$ & $10\%$\\
    \hline
       $a$ & $\$10$ & $\$10$ & $\$10$ \\
       $b$ & $\$10$ & $\$0$ & $\$11$\\
    \hline
       $c$ & $\$0$ & $\$10$ & $\$10$ \\
       $d$  & $\$0$ & $\$0$ & $\$11$ \\
    \hline
    \end{tabular}
    \caption{Classical Allais  problem}
    \label{tab:allais}
  \end{table}
Table~\ref{tab:allais} depicts a version of the classical Allais common consequence problem, where the modal behavior is choosing $a$ over $b$ and $d$ over $c$. If an individual has expected utility preferences, they will prefer $a$ over $b$ if and only if they prefer $c$ over $d$, regardless of their utility function over dollar amounts. Thus, the modal behavior is not compatible with expected utility. In contrast, this behavior can be rationalized by Min-Max-Mean preferences. Indeed, if an individual places a sufficiently high weight on the minimum, they will prefer $a$ to $b$ and $d$ to $c$.\footnote{If, for example, their utility function for money is simply the dollar amount, these choices are rationalized if the weight on the minimum is more than one-tenth of the weight on the maximum. } 

\begin{table}[ht]
    \centering
    \begin{tabular}{|c|c|c|c|}
    \hline
         Probability &  $1/3 $ & $1/3$ & $1/3$\\
    \hline
       $a$ & $\$10$ & $\$10$ & $\$10$ \\
       $b$ & $\$5$ & $\$5$ & $\$18$\\
       $c$ & $\$0$ & $\$10$ & $\$20$ \\
    \hline
    \end{tabular}
    \caption{Min-Max-Mean preferences allow for a choice of $b$}
    \label{tab:flexible_risk_prefs}
  \end{table}

If payoffs are interpreted as money rather than utilities, Min-Max-Mean preferences are flexible enough to explain choices that are not consistent with risk-averse or risk-seeking expected-utility preferences. For the lotteries depicted in Table~\ref{tab:flexible_risk_prefs}, any risk-averse individual will strictly prefer $a$ among  $\{a,b,c\}$, any risk-seeking individual would strictly prefer $c$, and a risk-neutral individual would be indifferent between $a$ and $c$, but strictly prefer both to $b$.\footnote{This holds because $a$ dominates $b$ and $c$ in the increasing concave order and $c$ dominates $a$ and $b$ in the increasing convex order.} In contrast, an individual with Min-Max-Mean preferences who places equal weight on the max and min (and sufficiently little weight on the mean) will rank $b$ the highest. When faced with binary choices, this decision maker would choose $b$ over $a$, exhibiting risk-loving behavior, and $b$ over $c$, exhibiting risk-averse behavior. As suggested by these examples, the combination of parsimony and flexibility of Min-Max-Mean preferences may prove useful in analyzing experimental data.

We conclude by noting that Min-Max-Mean preferences cannot, of course, explain all non-expected utility behavior. For example, they are consistent with expected utility behavior for a variant of Table~\ref{tab:allais} in which minima and maxima do not change across the two choices.

\section{Conclusion}

We have developed an axiomatic theory of how players bracket and evaluate strategic 
decisions, characterizing both classical solution concepts and new ones that accommodate non-expected-utility behavior, or mixed risk attitudes. Our analysis of bracketing reveals it as a unifying principle underlying Nash equilibrium and $\LQRE$, while also suggesting natural generalizations of 
these concepts. We discuss below three questions for future research.

In the definition of our solution concepts, it is hardwired that players randomize their actions independently.
A natural next step would be to consider solution concepts that allow for correlated actions, such as the set of correlated equilibria. 
Since the set of correlated equilibria satisfies a version of the bracketing axiom, extending our approach to correlated actions would potentially result in a new perspective on correlated equilibria and suggest their generalizations. The main challenge is finding the right analog of the monotonicity axioms.

One can also weaken the bracketing axiom by relaxing the independence requirement and only asking that if $p\in S(G)$ and $q\in S(H)$, then some mixed strategy profile $r$ with marginals $p$ and $q$ is in $S(G\otimes H)$. For example, such an axiom would allow for mixed LQRE, where the parameter $\lambda$ is sampled from a distribution over $\R_+$ that is fixed across games. 
We conjecture that such rules and their SRE cousins exhaust all the rules satisfying our axioms with this relaxed bracketing assumption. A positive resolution to this conjecture was established for one-player games in \cite*{sandomirskiy2025iid}.

Our results show that Nash equilibria,  $\LQRE$, and their generalizations arise under bracketing, as do some of their refinements. A complete 
characterization of which refinements permit bracketing remains open.
This question 
is particularly interesting for $\LQRE$, where the study of refinements has not been as extensive as in the Nash setting.

Bracketing may have different implications if restricted to particular classes of games, e.g., symmetric or zero-sum games.
Solution concepts that fail our axioms when applied 
to all games might satisfy them within these restricted domains, potentially yielding 
novel equilibrium notions tailored to particular strategic environments.

\bibliography{refs}

\newpage

\appendix

\section{Proof of Theorem~\ref{th_Nash QRE} }\label{ap_A}

We begin by introducing a number of definitions and establishing a key lemma. Recall that $S$ satisfies \textit{expectation-neutrality} if for any $G=(A,u)$ and $p \in S(G)$, if $\E[u_i(a,p_{-i})]=\E[u_i(b,p_{-i})]$, then $p_i(a) = p_i(b)$. A solution concept $S$ satisfies \emph{interiority} if every action in every game is played with positive probability, that is, for every game $G=(A,u)$, solution $p \in S(G)$, and player~$i$, we have $p_i(a_i)>0$ for each $a_i \in A_i$.

\begin{lemma}\label{lm_DSC_number}
  If $S$ permits bracketing and satisfies expectation-neutrality, interiority, expectation-monotonicity, and anonymity, then $S$ is a refinement of $\LQRE_\lambda$ for some $\lambda\geq 0$.
\end{lemma}
\begin{proof}[Proof of Lemma~\ref{lm_DSC_number}]

We first show that $S$ coincides with $\LQRE_\lambda$ on a particular class of games where all but one player have a single action. 
Fix a player~$i$ and consider for each $x \in \R$ the game $G_x = (A,u)$, where $A_i = \{a_0,a_x\}$, and payoffs of the player~$i$ are given by $u_i(a_0,a_{-i})=0$ and  $u_i(a_x,a_{-i})=x$, while all the other players have a single action and receive a payoff of zero for all action profiles. Fix for each $G_x$ a solution $p^x \in S(G_x)$, denote by $r_x = p^x_i(a_x)$ the probability that player~$i$ chooses $a_x$ under $p^x$ in the game $G_x$, and define 
$$
  f(x) = \log \left(\frac{r_x}{1-r_x}\right),
$$ which is well-defined by interiority.

We aim to demonstrate that $f(x)=\lambda\cdot x$ and so $S(G_x)=\LQRE_\lambda(G_x)$. 
By bracketing, for all $x,y \in \R$ it holds that $p^x \times p^y \times p^{x+y} \in S(G_x\otimes G_y\otimes G_{x+y})$. The actions $(a_x,a_y,a_0)$ and $(a_0,a_0,a_{x+y})$ yield the same payoffs, and hence by  expectation-neutrality, 
\begin{align*}
r_x r_y (1-r_{x+y}) = (1-r_x)(1-r_y)r_{x+y}.
\end{align*}
Rearranging and taking logs, we get $f(x+y)=f(x)+f(y)$, Cauchy's functional equation. Since $S$  satisfies expectation-monotonicity, applying bracketing to the composite game $G_x \otimes G_{x'}$ for $x>x'$ yields that $r_x$ is increasing in $x$, since the payoff for $(a_x,a_0)$ is larger than for $(a_{0},a_{x'})$, and so  
\begin{align*}
    r_x(1-r_{x'})\geq (1-r_x)r_{x'}.
\end{align*} Thus $f$ is nondecreasing. Since all monotone solutions to the Cauchy equation are linear, $f(x)=\lambda x$, for some $\lambda\geq 0$.

Fix any $G=(B,v)$ and $q\in S(G),$ with $b,c \in B_i$. Let $x=\E[v_i(b,q_{-i})]$ and $y=\E[v_i(c,q_{-i})]$. By bracketing, $s:=q \times p^{x-y}\in S(G \otimes G_{x-y})$. Let $w$ denote the payoff map for $G \otimes G_{x-y}$. Note that $\E[w_i((b,a_0),s_{-i})]=\E[w_i((c,a_{x-y}),s_{-i})]=x$. By expectation-neutrality, 
$$
    q_i(b)(1-r_{x-y})=q_i(c)r_{x-y}.
$$ 
Rearranging, we have $$
  \frac{q_i(b)}{q_i(c)}=\frac{r_{x-y}}{1-r_{x-y}}=\exp(f(x-y))=\exp(f(x))\exp(f(-y)).
$$

Since $c$ was arbitrary, $$q_i(b) \propto \exp (f(x))  =\exp (\lambda x) = \exp (\lambda \E[v_i(b,q_{-i})]), $$ for some $\lambda \geq 0$. By anonymity, this holds for all players~$i$.
\end{proof}

We are now ready to prove Theorem \ref{th_Nash QRE}.
\begin{proof}[Proof of Theorem~\ref{th_Nash QRE}]

Consider the game $G$ with action sets $A_1=\{h,\ell\}$ and $A_i=\{c\}$ for $i\neq 1$, and where the payoff for player~$1$ is $1$ when playing $h$ and $0$ when playing $\ell$. Let $p \in S(G)$. By Claim~\ref{cl_dominated},  if $p_1(h)=1$, then $S$ is a refinement of Nash.

For the remainder of this proof, suppose that $p_1(h)<1$. We will show $S$ is a refinement of $\text{LQRE}_{\lambda}$ for some $\lambda \geq 0$. First, we show that $S$ satisfies interiority. By anonymity, it is without loss of generality to suppose, toward a contradiction, that there is a game $H=(B,v)$ with $q \in S(H)$ and $a,b \in B_1$ such that $q_1(a) =0< q_1(b)$. Let $n>\E[v_1(b,q_{-1})]-\E[v_1(a,q_{-1})]$, and, as previously, consider that $r := q \times p^n\in S(H\otimes G^n)$. However, \[\E[u_1((h,\dots,h,a),r_{-1})]=n+\E[v_1(a,q_{-1})]>\E[v_1(b,q_{-1})]=\E[u_1((\ell,\dots,\ell,b),r_{-1})],\] while \[r_1(h,\dots,h,a)=(p_1(h))^nq_1(a)=0<r_1(\ell,\dots,\ell,b)=(p_1(\ell))^nq_1(b),\] violating expectation-monotonicity.

We now show that $S$ also satisfies expectation-neutrality. By anonymity, it is without loss of generality to suppose, toward a contradiction, that there is a game $H=(B,v)$ with $q \in S(H)$ and $a,b \in B_1$ such that $\E[v_1(a,q_{-1})]=\E[v_1(b,q_{-1})]$, while $q_1(a) < q_1(b)$. By interiority we may choose $n$ such that $\left(\frac{q_1(b)}{q_1(a)}\right)^n>\frac{p_1(h)}{p_1(\ell)}$. By bracketing, $r:=q^n \times p\in S(H^n\otimes G)$. However, $\E[u_1((h,a,\dots,a),r_{-1})]>\E[u_1((\ell,b,\dots,b),r_{-1})]$, while \[r_1(h,a,\dots,a)=(q_1(a))^np_1(h)<(q_1(b))^np_1(\ell)=r_1(\ell,b,\dots,b),\] violating expectation-monotonicity.

Since $S$ satisfies the hypotheses of Lemma \ref{lm_DSC_number}, there is a $\lambda \geq 0$, such that for any $i$, $G=(A,u),$ and $p \in S(G)$,
\[
p_i(a) \propto \exp(\lambda\E[u_i(a,p_{-i})]).
\]
\end{proof}

\section{Preliminary Results on Monotone Additive Statistics} \label{ap_MAS}

In this appendix we prove two results on monotone additive statistics that will be useful in a number of our proofs. See \S\ref{sec:SRE} for a discussion of monotone additive statistics and related notation.

\cite*{mu2024monotone} provide a characterization of monotone additive statistics on the domain of all compactly supported lotteries. However, lotteries originating from finite normal form games necessarily have a finite support. We show that the characterization of monotone additive statistics extends to those statistics that are defined on this restricted set of lotteries.

Let $\Delta_\fin$ be the set of all lotteries with finitely many real-valued outcomes. We denote by $\Delta_{\Q}\subset \Delta_\fin$ the set of all lotteries with a finite number of outcomes and rational-valued CDFs. Lotteries from $\Delta_{\Q}$ are convenient because of the following remark, which we use below to construct games that result in a prescribed payoff distribution of an action under minimal assumptions on the solution concept. 

\begin{remark}\label{rm_lottery_rv}
Any $X \in \Delta_{\mathbb{Q}}$ can be represented as a random variable defined on the probability space $(\Omega=\{1,\dots,m\},2^{\Omega},\nu)$, where $\nu$ is the uniform distribution, $m \in \N$, and $X \colon \Omega \to \R$.
\end{remark}

Our first step is to demonstrate that $\Delta_\Q$ is rich enough to approximate any compactly supported lottery in a monotone way, in particular, any lottery in $\Delta_\fin$. For a compactly supported lottery $X$ with CDF $F$, we denote by $\ubar X_n$ the lottery in $\Delta_\Q$ with CDF $\ubar F_n$ defined by $\ubar F_n(t)=\frac{1}{n!}\ceil{n!\cdot F(t)}$. Likewise, denote by $\bar X_n$ the lottery with CDF $\bar F_n$ defined by $\bar F_n(t)=\frac{1}{n!}\floor{n!\cdot F(t)}$. Note that for all $n$, 
\begin{equation}\label{eq_X_n_increasing}
    \ubar X_n \leq_{\text{FOSD}} \ubar X_{n+1} \leq_{\text{FOSD}} X \leq_{\text{FOSD}} \bar X_{n+1} \leq_{\text{FOSD}} \bar X_{n},
\end{equation} since first-order dominance is equivalent to having a lower CDF (pointwise). Note also that if $X \in \Delta_\Q$ then $\ubar F_n = F  = \bar F_n$ for all $n$ large enough, since $F$ only takes finitely many rational values, and thus $n!\cdot F(t)$ is an integer once $n!$ is larger than these values' lowest common denominator.

We seek to show that the sequences of lotteries $(\ubar X_n)_n$ and $(\bar X_n)_n$ provide an approximation for $X$. To this end we require a technical lemma about stochastic dominance in large numbers. Recall that we identify a lottery $X$ with a distribution over $\R$, and write $X+Y$ for the lottery corresponding to the sum of outcomes independently sampled from $X$ and $Y$. For $m \in \N$, let $X^m$ denote the sum of $m$ independent copies of $X$.
\begin{definition}
     Let $X$ and $Y$ be compactly supported lotteries. Say \textit{$X$ dominates $Y$ in large numbers}, denoted $X >_L Y$, if there exists $M \in \N$ such that for all $m \geq M$, $X^{m} \g Y^{m}$.
\end{definition}
The following result is due to \cite*{aubrun2007catalytic}:
\begin{proposition}[Aubrun and Nechita]\label{lm_large_numbers aubrun nechita}
    Let $X$ and $Y$ be compactly supported lotteries with $K_a[X] > K_a[Y]$ for all $a \in \overline{\R}$. Then $X >_L Y$.
\end{proposition} 
If $\Phi\colon \Delta_\Q \to \R$ is monotone and additive and $X,Y \in \Delta_\Q$ then $X>_L Y$ implies that $\Phi(X)\ge \Phi(Y)$, since
\begin{align*}
    \Phi(X) = \frac{1}{m}\Phi(X^m) \geq \frac{1}{m}\Phi(Y^m) = \Phi(Y),
\end{align*}
where the first and last equalities follow from additivity, and the inequality holds for $m$ large enough since $X >_L Y$, and by the monotonicity of $\Phi$.

\begin{lemma}\label{lm_lim_Phi_converge}
    Let $\Phi \colon \Delta_\Q \to \R$ be a monotone additive statistic and $X\in\Delta_\fin$. Then  \begin{equation*}\
         \lim_n \Phi[\ubar X_n] = \lim_n \Phi[\bar X_n].
    \end{equation*}
\end{lemma}
\begin{proof}
    Let $F$ denote the CDF of $X$. Since $X\in\Delta_\fin$, there is a closed interval $I\subset \R$ containing all the values $X$ that takes with positive probability. Denote by $\Delta_I$ the set of all lotteries supported on~$I$. Lotteries $X$, $\bar X_n$, and $\ubar X_n$ belong to $\Delta_I$.
    
    Define a sequence of functions $(g_n)_n$ with $g_n\colon \overline \R\to \R$ by $g_n(a)=K_a[\bar X_n]-K_a[\ubar X_n]$. 
    For any $a\in \R$, we have $\lim_n g_n(a)=0$, since $\bar X_n$ and $\ubar X_n$  converge to $X$ weakly, and $K_a$, considered as a mapping $ \Delta_I \to \R$, is continuous in the weak topology. 
    Consider now $a=\pm\infty$. Recall that $K_{-\infty}$ and $K_{\infty}$ are the leftmost and the rightmost points of the support, respectively. Therefore, $\lim_n g_n(a)=0$ also for $a=\pm\infty$. Indeed,
     $K_{-\infty}[\ubar X_n]=K_{-\infty}[X]$ for all $n$ and $K_{-\infty}[\bar X_n]=\min \{t \mid \bar F_n(t) > 0\}$ converges to $K_{-\infty}[X]=\min \{t \mid F(t) > 0\}$.
    Thus $\lim_n g_n(-\infty)=0$, and an analogous argument shows that $\lim_n g_n(\infty)=0$. 
     
    As each $K_a$ is monotone with respect to first-order dominance, $g_{n+1}(a)\leq g_n(a)$ for any~$a$ and~$n$. Note that each $g_n$ is a continuous function of $a$. Thus, $(g_n)_n$ is a decreasing sequence of continuous functions such that $g_n(a)\to 0$ for each $a$ from the compact set $\overline \R$.
    By Dini's theorem, monotone pointwise convergence to a continuous function on a compact set implies uniform convergence.    
    Thus, for each $\varepsilon > 0$, there exists $M \in \N$ such that for $n\geq M$, 
 \[K_a[\ubar X_n +\varepsilon]=K_a[\ubar X_n] +\varepsilon > K_a[\bar X_n],\]
for all $a \in \overline \R$. By Proposition~\ref{lm_large_numbers aubrun nechita} and the remark following it, $\Phi[\ubar X_n] + \varepsilon \geq \Phi[\bar X_n]$. 
Taking the limit as $\varepsilon$ goes to zero, we conclude that $\lim_n \Phi[\ubar X_n] = \lim_n \Phi[\bar X_n]$.

\end{proof}

\begin{lemma}\label{lm_rational_MAS}
 Let $\Phi \colon \Delta_{\mathbb Q} \to \R$ be a monotone additive statistic. Then $$\Phi[X]=\int_{\overline{\R}}K_a[X]\,\dd \mu(a)$$ for some Borel probability measure $\mu$ on $\overline{\R}$.
\end{lemma}

\begin{proof}[Proof of Lemma~\ref{lm_rational_MAS}]
Let $\Phi \colon \Delta_{\mathbb{Q}} \rightarrow \R$ be a monotone additive statistic and fix a lottery~$X$ that is compactly supported. 
As above, we denote by $\ubar X_n$ the lottery in $\Delta_\Q$ with CDF $\ubar F_n(t)=\frac{1}{n!}\ceil{n!\cdot F(t)}$. 
Define the real-valued function $\Psi$ on the set of compactly supported lotteries by $\Psi[X]=\lim_{n \rightarrow \infty}\Phi[\ubar X_n]$. By \eqref{eq_X_n_increasing}, $( \ubar X_n)_n$ is an increasing sequence in terms of first-order dominance, so $\Psi[X]\geq \Phi[\ubar X_n]$ for all $n$. For $X \in \Delta_\Q$, $\ubar X_n$ and $X$ have the same distribution for all $n$ large enough (see the remark after \eqref{eq_X_n_increasing}), $\Psi[X] =\lim_n \Phi[\ubar X_n]=\Phi[X]$, i.e., $\Psi$ extends $\Phi$. 

Let $X,Y$ be compactly supported lotteries. If $X \geq_{\text{FOSD}} Y$, then $\ubar X_n  \geq_{\text{FOSD}} \ubar Y_n$ for all~$n$, so $\Psi[X] \geq \Psi[Y]$, and $\Psi$ is monotone. Moreover, if $X$ and $Y$ are independent, then the minima and maxima of $\ubar X_n + \ubar Y_n$ and $(\underline{X + Y})_n$ converge to those of $X + Y$. Additionally, they both converge in distribution to $X+Y$, so,   by the same argument as in the proof of Lemma~\ref{lm_lim_Phi_converge},  
$\lim_n\Phi[\ubar X_n + \ubar Y_n]=\lim_n\Phi[(\underline{X + Y})_n]$. It thus holds that \begin{align*}
  \Psi[X+Y]&= \lim_{n \to \infty}\Phi[(\underline{X + Y})_n]= \lim_{n \rightarrow \infty}\Phi[\ubar{X}_n + \ubar{Y}_n ]\\
  &= \lim_{n \rightarrow \infty}\Phi[\ubar{X}_n] + \Phi[\ubar{Y}_n] = \Psi[X]+\Psi[Y].
\end{align*}

Hence, by the characterization of \cite*{mu2024monotone}, $\Psi[X]=\int_{\overline{\R}}K_a[X]\,\dd \mu(a)$ for some Borel probability measure $\mu$ on $\overline{\R}$. On its domain $\Delta_\Q$, the statistic~$\Phi$ coincides with $\Psi$ and thus admits the same representation. 
\end{proof}

\begin{lemma}\label{lm_additive_MAS}
    Let $\Phi \colon \Delta_{\mathbb Q} \to \R$ be a monotone statistic. If $\Phi$ is additive for all lotteries (not necessarily independent), then $\Phi$ is the expectation. 
\end{lemma}
\begin{proof}[Proof of Lemma \ref{lm_additive_MAS}]
As in Remark~\ref{rm_lottery_rv}, we consider each finite probability space $(\Omega= \{1,\dots,m\}, 2^\Omega, \nu)$, where $\nu$ is the uniform distribution on $\Omega$. Each $X \in \Delta_{\mathbb Q}$ can be represented as a random variable $X \colon \Omega \to \R$ for a large enough $m$. We require that $\Phi[X+Y] = \Phi[X]+\Phi[Y]$ for any random variables $X,Y \colon \Omega\to \R$. 

For $X \colon \Omega \to \R$, we may write $X= \sum_\omega I(\omega) X(\omega)$, where $I$ denotes the indicator function. For each $\omega \in \R$, define $f_\omega \colon \R \to \R$ by $f_\omega(x) = \Phi[I(\omega)\cdot x]$, so $\Phi[X] = \sum_\omega f_\omega(X(\omega))$. It follows that each $f_\omega$ is a monotone additive function and is therefore linear. Thus there is $Z \in \R^\Omega$ such that $\Phi[X]=Z\cdot X$ for all $X$. Since $\Phi$ only depends on the distribution of $X$ and $\mu$ is uniform, $\Phi[X]=\Phi[X \circ \pi] $ for any permutation $\pi \colon \Omega \to \Omega$, which is only possible for constant $Z$. Finally, since $\Phi$ is a statistic, it maps any constant random variable to its value, so $Z(\omega)=\frac{1}{m}$ for all $\omega$. We have thus shown that for any $X\colon \Omega \to\R$, $\Phi[X]=\frac{1}{|\Omega|}\sum_\omega X(\omega)=\E[X]$.
\end{proof}

\section{Proof of Proposition \ref{pr_Nash_Phi_Exist}}\label{ap_pr_1}
We show the two parts of the proposition in two separate claims: Claim~\ref{clm_lqre} and Claim~\ref{clm_nash}.

Given $\lambda\geq 0$ and $\Phi = \int K_t\,\dd\mu(t)$, we prove in Claim~\ref{clm_lqre} that every game $G=(A,u)$ has an $\LP$ equilibrium. A natural approach would be to follow the proof of the existence of quantal response equilibria by defining the quantal response operator $T \colon \prod_i \Delta A_i \to \prod_i \Delta A_i$ by $T_i(p)(a_i) \propto \exp(\lambda \Phi[u_i(a_i,p_{-i})])$ and applying a fixed point theorem. The issue is that when $\mu$ has positive mass at $t=\infty$ or $t=-\infty$, then $T$ is not continuous, and so we cannot apply Brouwer's fixed point theorem. 

To overcome this issue, we define a closely related, continuous operator $T'$, apply the fixed point theorem to it, and then show that this is also a fixed point of $T$, and hence an equilibrium. To define $T'$, fix  $i$, $a_i\in A_i$, and $p\in \prod_i\Delta A_i$. Let $L \colon \overline{\R} \times A_i \times \prod_{j \neq i}\Delta A_j \to \R$ be given by
    \begin{align}
        \label{eq_Lt}
      L(t,a_i,p_{-i}) = 
      \begin{cases}
          K_t[u_i(a_i,p_{-i})]& \text{if } t\in \R \\
          \min_{a_{-i}}u_i(a_i,a_{-i})& \text{if } t = -\infty \\
          \max_{a_{-i}}u_i(a_i,a_{-i})& \text{if } t = +\infty.
      \end{cases}
    \end{align}
    That is, when $t \in \R$, $L(t,a_i,p_{-i})$ is equal to the monotone additive statistic $K_t$, evaluated on the lottery that player~$i$ gets when playing $a_i$ and when the rest of the players play $p_{-i}$. For $t \in \{-\infty,+\infty\}$, 
    $L(t,a_i,p_{-i})$ is independent of $p_{-i}$, and returns the minimum or maximum payoff that the action $a_i$ can yield.  Let $\Phi'[a_i,p_{-i}]=\int_{\overline{\R}} L(t,a_i,p_{-i})\,\dd \mu(t)$. Note that if $p_{-i}$ is totally mixed then $\Phi'[a_i,p_{-i}]=\Phi[u_i(a_i,p_{-i})]$. Note also that if $\mu(\{-\infty,+\infty\})=0$ then $\Phi'[a_i,p_{-i}]=\Phi[u_i(a_i,p_{-i})]$ for all $p_{-i}$. 

    Define $T' \colon \prod_i \Delta A_i \to \prod_i \Delta A_i$ by
    \begin{align}
        \label{eq_Tprime}
        T_i'(p)(a_i) \propto \exp(\lambda \Phi'[a_i,p_{-i}]).
    \end{align}
    To prove the existence of $\LP$ equilibria, we show that $T'$ is continuous, and that its fixed points coincide with those of $T$, and hence are equilibria.
\begin{claim}
    \label{clm_lqre}
    Let $\Phi=\int K_a\, \dd \mu(a)$ be a monotone additive statistic. Then
    there is an $\LQRE_{\lambda\Phi}$ equilibrium for every game, for every $\lambda \geq 0$.

\end{claim}
\begin{proof}
    Fix $\lambda \geq 0$, $\Phi=\int K_a\, \dd \mu(a)$, and let $G=(A,u)$. Define $L$ and $T'$ as in \eqref{eq_Lt} and \eqref{eq_Tprime}. 

    For $t \in \R$, the map $p_{-i} \mapsto L(t, a_i,p_{-i})$ varies continuously in $p_{-i}$. It is also (trivially) continuous when $t\in \{-\infty,+\infty\}$, since then it does not depend on $p_{-i}$. It follows that $L(t,a_i,p_{-i})$ is continuous in $p_{-i}$ for all $t \in \overline{\R}$ and $a_i \in A_i$.

    We show that $T'$ has a fixed-point. Since $L(t,a_i,p_{-i})$ is continuous in $p$ for all $t \in \overline{\R}$ and 
    \[
      |L(t,a_i,p_{-i})| \leq \max_{a_{-i}} |u_i(a_i,a_{-i})|,
    \] 
    by the dominated convergence theorem and the continuity of $L$,  
    \begin{align*}
        \lim_{p_n\to p}\Phi'[a_i,{(p_n)}_{-i}] &=\lim_{p_n\to p}\int L\left(t,a_i,{(p_{n})}_{-i}\right) \,\dd \mu(t) \\
        &= \int \lim_{p_n\to p} L\left(t,a_i,{(p_{n})}_{-i}\right) \,\dd \mu(t)= \int L\left(t,a_i,p_{-i}\right) \,\dd \mu(t) = \Phi'[a_i,p_{-i}],
    \end{align*}
    hence, $T'$ is continuous in $p$. Since $\prod_i \Delta A_i$ is convex and compact, $T'$ has a fixed-point $q^*$ by Brouwer's fixed-point theorem. Since $T'$ maps every mixed strategy profile to a totally mixed strategy profile, $q^*$ must be totally mixed. We thus have 
    \[
    q^*_i(a_i) \propto \exp(\lambda \Phi'[a_i,q^*_{-i}])=\exp(\lambda \Phi[u_i(a_i,q^*_{-i})]),
    \]
    as $\Phi$ and $\Phi'$ agree when $i$'s opponents play totally mixed strategy profiles. Hence we have shown that an $\LQRE_{\lambda\Phi}$ equilibrium exists.
\end{proof}

Next, we show the second part of the proposition:
\begin{claim}
    \label{clm_nash}
    Let $\Phi=\int K_a\, \dd \mu(a)$ be a monotone additive statistic. Then there is a $\Nash_\Phi$ equilibrium for every game if and only if  $\mu(\{-\infty, +\infty\})=0$.
\end{claim}

\begin{proof}
    The existence of $\Nash_{\Phi}$ equilibria for $\mu(\{-\infty,+\infty\})=0$ follows from Kakutani's fixed-point theorem since the best response correspondence is upper hemicontinuous for such $\mu$. Alternatively, we can show the existence of $\Nash_{\Phi}$ as a limit point of $\LP$.

    Next, we demonstrate how to construct a game with no $\Nash_\Phi$ equilibrium when $\mu$ places a positive weight on the minimum or maximum. For such a $\Phi$, let $\varepsilon = \mu(-\infty)+\mu(+\infty)$ and consider the game in table \ref{tab:no_min_max_eq}. 

        \begin{table}[h]
    \centering
    \begin{tabular}{c|c|c}
            & $a_2$&$b_2$ \\ \hline
         $a_1$&  $(1+\frac{1}{\varepsilon},0)$& $(0,1)$\\
         $b_1$&  $(-\frac{1}{\varepsilon},1)$& $(1,0)$\\
    \end{tabular}
    \caption{Variant of matching pennies for which extremal $\Nash_\Phi$ equilibria do not exist.}
    \label{tab:no_min_max_eq}
\end{table}

    Since pure $\Nash_\Phi$ equilibria coincide with pure $\Nash$ equilibria for all $\Phi$, it is easy to see that the game has no pure equilibria. Likewise, there are no equilibria where either player plays a pure strategy, since the best responses to pure strategies in this game are pure for all $\Phi$. In particular, any supposed $\Nash_\Phi$ equilibrium $q$ would have player~2 playing a totally mixed strategy. We thus have

    \begin{align*}
        \Phi[u_1(a_1,q_2)]&-\Phi[u_1(b_1,q_2)]=(\mu(-\infty)+\mu(+\infty))\cdot \frac{1}{\varepsilon} \\ 
        &+ \int_{\R}\underbrace{K_t[u_1(a_1,q_2)]}_{\text{nonnegative}}
\,\dd \mu(t) - \int_{\R}\underbrace{K_t[u_1(b_1,q_2)]}_{\leq 1}\,\dd \mu(t) \geq 1 - \mu(\R)\cdot 1=\varepsilon>0.
    \end{align*}
    This contradicts the assumption that $q$ is a totally mixed $\Nash_\Phi$ equilibrium, which would require that $\Phi[u_1(a_1,q_2)]=\Phi[u_1(b_1,q_2)]$.

\end{proof}

\section{Proof of Proposition \ref{lm_Ghl}}\label{ap_distr_neutrality}

Consider a solution concept $S$ that permits bracketing and satisfies distribution-monotonicity, and
anonymity. For such $S$, Proposition~\ref{lm_Ghl} claims that $S$ is either a refinement of ordinal-Nash or a refinement of ordinal-QRE. We prove this proposition by showing the following dichotomy:

\begin{enumerate}
    \item If players never play strictly dominated actions, then $S$ is a refinement of ordinal-Nash.
    \item If players play a strictly dominated action in some game, then $S$ is a refinement of ordinal-QRE.
\end{enumerate}
We consider each of the two cases separately in Claims~\ref{clm_distribution-mono1} and~\ref{clm_distribution-mono2} below. The proposition follows directly from them. 
To prove Claim~\ref{clm_distribution-mono1}, we first provide a condition under which adding independent lotteries to lotteries ranked with respect to first-order stochastic dominance preserves the dominance ranking.

\begin{lemma}\label{lm_K_trick}
    Let $X,Y,A,$ and  $B$ be compactly supported lotteries with $X \g Y$, $\max(A)>\max(B)$ and $\min(A)>\min(B)$. Then there exist $m,n \in \N$ such that $$X^m + A^n \g Y^m +B^n.$$
\end{lemma}
\begin{proof}
    Since $X \g Y$, we have $K_a[X] > K_a[Y]$ for all $a \in \R$ and $K_a[X] \geq K_a[Y]$ for $a = \pm\infty$. Moreover, $K_a[A] > K_a[B]$ for $a = \pm\infty$. By continuity of $K_a$ in $a$, there exists $M >0$ such that $K_a[A] > K_a[B]$ for all $a \in \overline\R\setminus[-M,M]$. Consider $t=\min_{a \in [-M,M]}(K_a[X]-K_a[Y])>0$ and $s=\min_{a \in [-M,M]}(K_a[A]-K_a[B])$. Choose $d \in \N$ such that $d\cdot t+s>0$. It thus follows from Proposition~\ref{lm_large_numbers aubrun nechita} that $X^d+A >_L Y^d+B$. The result follows from the definition of $>_L$.
\end{proof}

To make use of the above lemma, we construct a variant of matching pennies such that any solution must involve someone playing, with positive probability, an action that generates a lottery with a lower max and min than its alternative.

    \begin{table}[h]
    \centering
    \begin{tabular}{c|c|c}
            & $a_2$&$b_2$ \\ \hline
         $a_1$&  $(2,0)$& $(0,1)$\\
         $b_1$&  $(-1,1)$& $(1,0)$\\
    \end{tabular}
    \caption{Variant of matching pennies}
    \label{tab:VMP_low_maxmin}
\end{table}
\begin{lemma}
\label{lm_min_max_H}
There exists a game $H$ for $n \geq 2$  players, with integral payoffs, and such that for all mixed strategy profiles $p$, there is a player~$i$ and actions $a_i,b_i \in A_i$ such that $p_i(b_i)>0$ and $\min[u_i(b_i,p_{-i})] < \min[u_i(a_i,p_{-i})]$      and $\max[u_i(b_i,p_{-i})] < \max[u_i(a_i,p_{-i})]$.
\end{lemma}
\begin{proof}
    Let $H$ be the game in which players 1 and 2 play the game described in Table~\ref{tab:VMP_low_maxmin}, and the remaining players' actions do not affect the payoffs of players 1 and 2.

    Since this game has no pure Nash equilibria, any pure strategy profile has the desired property for whichever of players 1 and 2 has a profitable deviation. It is easy to verify that if $p$ has the property that one of the first two players totally mixes and the other plays a pure strategy, then $p$ has the desired property with respect to the player who is mixing. 

    Finally, if both players 1 and 2 play 
    totally mixed strategies, then 
    \begin{align*}
        \min[u_1(b_1,p_{-1})]=-1 &< \min[u_1(a_1,p_{-1})] = 0, \\
      \max[u_1(b_1,p_{-1})]=1 &< \max[u_1(a_1,p_{-1})]=2.
     \end{align*}
     Thus $p$ has the desired property with respect to player~1.
\end{proof}

\begin{claim}
\label{clm_distribution-mono1}
    Suppose $S$ permits bracketing and satisfies distribution-monotonicity. Assume also that players never play strictly dominated actions. Then $S$ is a refinement of ordinal-Nash.
\end{claim}
\begin{proof}
Let $H$ be a game that satisfies the property whose existence is guaranteed 
by Lemma~\ref{lm_min_max_H} (e.g., the one described in Table~\ref{tab:VMP_low_maxmin}). Fix any $p \in S(H)$. By the defining property of $H$ there is a player~$i$ such that $i$ plays with positive probability an action $\underline a_i$ that yields a lottery with a lower max and min than its alternative $\overline{a}_i$. 

Consider a game $F=(C,w)$ where $C_i=\{a_{0.5},a_0\}$ and payoffs are given by $w_i(a_{0.5},\,\cdot\,)=0.5$ and $w_i(a_0,\,\cdot\,)=0$. Note that $q_i(a_{0.5})=1$ for any solution $q \in S(F)$, by the assumption that strictly dominated actions are never played.
Consider the composite game $H\otimes F$. 
By bracketing, $p\times q$ is one of its solutions. It puts positive weight on $(\underline a_i,a_{0.5})$ since $p_i(\underline a_i)\cdot q_i(a_{0.5})>0$, while $(\overline a_i, a_{0})$ has probability zero as $p_i(\overline a_i)\cdot q_i(a_{0})=0$.
Since the payoffs in $H$ are integral, a difference in the max or min of lotteries generated by a player's actions is at least~$1$. Consequently, the payoff distribution of $(\overline a_i, a_{0})$ has strictly higher maximum and minimum than those of $(\underline a_i,a_{0.5})$: 
\begin{align*}
    \max (v_i(\overline{a}_i,p_{-i})+w_i(a_0,q_{-i}))&>\max(v_i(\underline{a}_i,p_{-i})+w_i(a_{0.5},q_{-i})),\\
    \min (v_i(\overline{a}_i,p_{-i})+w_i(a_0,q_{-i}))&>\min(v_i(\underline{a}_i,p_{-i})+w_i(a_{0.5},q_{-i})).
\end{align*}

Let $G = (A,u)$ be an arbitrary game with $a_{Y},a_{Z} \in A_i$ and $r \in S(G)$ such that $u_i(a_Y,r_{-i}) \g u_i(a_Z,r_{-i})$. We need to show that $r_i(a_Z)=0$. Indeed, by Lemma~\ref{lm_K_trick} there exist $m,n \in \N$ such that \[u_i(a_Y,r_{-i})^m + (v_i(\overline{a}_i,p_{-i})+w_i(a_0,q_{-i}))^n\g u_i(a_Z,r_{-i})^m + (v_i(\underline{a}_i,p_{-i})+w_i(a_{0.5},q_{-i}))^n.\]
By bracketing, $r^m \times p^n \times q^n \in S(G^m \otimes H^n \otimes F^n)$. Since $q_i(a_{0})=0$ and $p_i(\underline a_i),q_i(a_{0.5})>0$, by distribution-monotonicity, it must be that $r_i(a_Z)=0$.
\end{proof}

\begin{claim}
\label{clm_distribution-mono2}
    Suppose $S$ permits bracketing and satisfies distribution-monotonicity, and anonymity. Also, assume that there is a game $G_{D}$, a player~$i$, and a solution $p \in S(G_D)$ in which player~$i$ plays a dominated action with positive probability. Then $S$ is a refinement of ordinal-QRE.
\end{claim}
\begin{proof}
By definition, there are actions $a_h,a_\ell$ in $G_{D}$ and $D > 0$ such that $p_i(a_\ell)>0$ and $u_i(a_h,a_{-i}) \geq u_i(a_\ell,a_{-i})+D$ for all $a_{-i}$. 
We show that $S$ must then satisfy interiority. Let $G=(A,u)$ be any other game, and consider some $a \in A_i$ and $q \in S(G)$. To prove interiority, we need to show that $q_i(a)>0$. Pick some $b$ such that $q_i(b)>0$. Consider $m \in \N$ such that $m D>\max [u_i(b,q_{-i})]-\min [u_i(a,q_{-i})]$. By  bracketing, $r=q \times p^m$ is in $S(G_{D}^m \otimes G)$. Denote  the payoff map of $G_{D}^m \otimes G$ by~$v$, and note that \[v_i(a_h,\dots,a_h,a,r_{-i})>_{\text{FOSD}} v_i(a_\ell,\dots,a_\ell,b,r_{-i}).\] 
Therefore, by distribution-monotonicity and  bracketing,
 \[r_i(a_h,\dots,a_h,a)\geq r_i(a_\ell,\dots,a_\ell,b)>0,\]
and so $q_i$ is totally mixed. We conclude that $S$ satisfies interiority.

We next show that such an $S$ must satisfy distribution-neutrality. Toward a contradiction, suppose that distribution-neutrality is violated in a game $G=(A,u)$.  By anonymity, we can assume that the violation occurs for the same player~$i$ that played the dominated action $a_\ell$ in $G_D$. Hence, there is $q \in S(G)$ and $a,b \in A_i$, such that $u_i(a,q_{-i})=u_i(b,q_{-i})$, while $q_i(a) < q_i(b)$. By assumption $p\in S(G_{D})$ satisfies $p_i(a_\ell)>0$. 
Pick $m \in \N$ such that $\left(\frac{q_i(b)}{q_i(a)}\right)^m>\frac{p_i(a_h)}{p_i(a_\ell)}$. By  bracketing, $r=q^m \times p\in S(G^m\otimes G_{D})$. Let $v$ denote the payoff map of $G^m\otimes G_{D}$, and note that $v_i(a,\dots,a,a_h,r_{-i})>_{\text{FOSD}}v_i(b,\dots,b,a_\ell,r_{-i})$, while \[r_i(a,\dots,a,a_h)<r_i(b,\dots,b,a_\ell),\] violating distribution-monotonicity. This contradiction implies that $S$ satisfies distribution-neutrality.
\end{proof}

Claim~\ref{clm_distribution-mono1} and Claim~\ref{clm_distribution-mono2} together immediately imply Proposition~\ref{lm_Ghl}.

\section{Test Games}\label{ap_test}
In this appendix, we study the test games needed in the proof of
Theorem~\ref{th_sre} to elicit players’ preferences over payoff
lotteries. This proof will require two classes of such games, corresponding to the two cases of Proposition~\ref{lm_Ghl}. The first, $G_{r,x,\varepsilon}$, will be useful in the first case, where players play an ordinal-Nash. The second, $H_{r,x}$, will be useful for ordinal-QRE.

Recall first the definition of the game $G_{r,x,\varepsilon}$ from \eqref{eq:G_rXeps}.\footnote{While this is an $n \geq 2$ player game, we write the utilities of the first two players as functions of the actions of the first two players only, suppressing the actions of the rest.} The next lemma proves the properties of $G_{r,x,\varepsilon}$ that make it useful as a test game. Namely, that when player~1 chooses $a_r$ they receive (approximately) $r$, and that in any ordinal-Nash player~2 always mixes uniformly, so that when player~1 chooses $a_x$, they receive a lottery distributed as the uniform distribution over $\{x_1,\ldots,x_m\}$.
\begin{lemma}\label{lm_lotteries_games_FOSD_undominated}

    For each $r \in \R$, every nonconstant $x \in \R^m$, and $\varepsilon > 0$, the game $G_{r,x,\varepsilon}$ has the following properties: 
    \begin{enumerate}
    \item $|u_1((a_r,\pi),a_2) - r|\leq \varepsilon \norm{x}_\infty$ for all $(a_r, \pi) \in A_1$ and $a_2 \in A_2$;
    \item $p_2$ is the uniform distribution over $A_2$ in any ordinal-Nash $p$.
    \end{enumerate}

\end{lemma}
\begin{proof}[Proof of Lemma~\ref{lm_lotteries_games_FOSD_undominated}]
Fix a nonconstant $x \in \R^m$, $r \in \R$, and $\varepsilon > 0$, and consider the game~$G_{r,x,\varepsilon}$. 

That (1) holds is immediate from the definition of the game. To show (2), suppose that $p$ is a strategy profile such that every first-order dominated action is played with probability zero. We need to show that $p_2$ is uniform over $A_2$.

Given a permutation $\pi$ of $\{1,\ldots,m\}$ denote by $x \circ \pi$ the vector $(x_{\pi(1)},\ldots,x_{\pi(m)}) \in \R^m$. First note that for all $\pi$ chosen by player~1, $x \circ \pi$ must be weakly increasing in $p_2$, i.e., for $s,t \in A_2$, $p_2(s)>p_2(t)\implies x_{\pi(s)} \geq x_{\pi(t)}$.\footnote{In the card game interpretation of this game, this statement means that if player~2 chooses the first card with probability strictly higher than the second card, then player~$1$ will order the cards so that the first card shows a higher payoff than the second.} Indeed, if, for some $s,t \in A_2$,  $p_2(s) > p_2(t)$ while $x_{\pi(s)}<x_{\pi(t)}$, then $\pi$ is first-order dominated by $\pi'$ which coincides with $\pi$ except on $\{s,t\}$, where $\pi'(s)=\pi(t)$ and $\pi'(t)=\pi(s)$. 

Fix any $\pi$ played with positive probability by player~1, and suppose, for the sake of contradiction, that $p_2$ is not uniform. Since $x$ is nonconstant, we claim that there must be $s,t \in A_2$ with $p_2(s) > p_2(t)$ and $x_{\pi(s)}>x_{\pi(t)}$. Indeed, fix any $s,t \in A_2$ such that $p_2(s) > p_2(t)$. If $x_{\pi(s)}>x_{\pi(t)}$, we are done. Suppose then that $x_{\pi(s)}=x_{\pi(t)}=c$. Since $x\circ \pi$ is nonconstant, there is an $h \in A_2$ such that $x_{\pi(h)}\neq c$. Consider the case where $x_{\pi(h)}>c$. Then $p_2(h)\geq p_2(s)$, since  $x \circ \pi$ is weakly increasing in $p_2$. We thus have $p_2(h) > p_2(t)$ and $x_{\pi(h)}>x_{\pi(t)}$, as desired. An identical argument works for the case where $x_{\pi(h)}<c$. 

We thus have that if $p_2$ is not uniform then there are $s,t \in A_2$ with $p_2(s) > p_2(t)$ and $x_{\pi(s)}>x_{\pi(t)}$.
We claim that, moreover, for all $\pi'$ played with positive probability, $x_{\pi'(s)}\geq x_{\pi'(t)}$, by the weakly increasing property of $x \circ \pi'$. Thus $s$ is first-order dominated by $t$ for player~$2$, a contradiction.
\end{proof}

We next construct a class of test games in which players must evaluate a rich set of lotteries against sure things under distribution-neutrality. This setting is much simpler than the previous, since it is easier to guarantee that players mix uniformly.

For $x\in \R^m$, we define $H_{r,x}=(B,v)$ by $B_1=\{b_r,b_x\}$, $B_2 = \{1,\ldots,m\}$,
\begin{align}
\label{eq_HrX}
    v_1(b_r,b_2)=r \quad\text{ and } \quad v_1(b_x,b_2)=x_{b_2}
\end{align}
for all $b_2 \in B_2$. Let $v_i$ be identically zero for all $i \neq 1$. As with $G_{r,x,\varepsilon}$, we write utilities as just functions of the first two players, since they do not depend on the actions of the rest.

The following lemma summarizes the key features of $H_{r,x}$:
\begin{lemma}\label{lm_lotteries_games}
For each $r \in \R$ and $x \in \R^m$, the game $H_{r,x}$ has the following properties: 
    \begin{enumerate}
    \item The action $b_r$ results in a deterministic payoff of $r$ to player~$1$, i.e., $v_1(b_r,p_2)$ is a degenerate lottery which yields~$r$;
\item $p_2$ is the uniform distribution over $B_2$ in any ordinal-QRE $p$.
\end{enumerate}    
\end{lemma}

\begin{proof}[Proof of Lemma~\ref{lm_lotteries_games}]
 Clearly $v_1(b_r,p_2)=r$, and by distribution-neutrality, $p_2$ is uniform.

\end{proof}

\section{Proof of Theorem~\ref{th_sre}}\label{ap_B}
Given a lottery $X \in \Delta_{\Q}$, we can find $m > 0$ and $x \in \R^m$ such that $X$ is the uniform distribution over $(x_1,\ldots,x_m)$. We define the game $G_{r,X,\varepsilon}$ to be equal to the game $G_{r,x,\varepsilon}$, for some canonical choice of such $x$ (e.g., the one with minimal $m$ and non-increasing components). We also write $a_X$ for the action $a_x$. The game $H_{r,X}$ is defined analogously.

The following lemma shows how certainty equivalents may be deduced from player~1's mixing probabilities in the games $G_{r,X,\varepsilon}$ under the assumption that players play an ordinal-Nash. We will refer to $X \in \Delta_\Q$ with the understanding that $X$ is a random variable as in Remark~\ref{rm_lottery_rv}, so that $G_{r,X,\varepsilon}$ is a well-defined game. 
\begin{lemma}\label{lm_phiep}
Suppose $S$ permits bracketing and is a refinement of ordinal-Nash. Define $\Phi_\varepsilon \colon \Delta_\Q \to \R$ by 
\begin{align}
    \label{eq_phi_eps}
    \Phi_\varepsilon[X]=\sup \{r\in \R \mid \exists p \in S(G_{r,X,\varepsilon}), \exists \pi \colon A_2 \to A_2 \text{ with } p_1(a_X,\pi)>0\}.
\end{align}
Then the limit $\Phi[X] = \lim_{\varepsilon \to 0}\Phi_\varepsilon[X]$ exists and is a monotone additive statistic.

\end{lemma}

\begin{proof}
We first show that the limit exists. Suppose, for the sake of contradiction that $\linf_{\varepsilon \to 0} \Phi_\varepsilon[X] < \lsup_{\varepsilon \to 0} \Phi_\varepsilon[X]$ for some $X \in \Delta_\Q$. There then exist $\delta > 0$ and $c > 0$ such that for any $\varepsilon < \delta$, there are $\varepsilon_1,\varepsilon_2<\varepsilon$ with $\Phi_{\varepsilon_1}[X] + c < \Phi_{\varepsilon_2}[X]$. Thus, there exist $(A,u)=G_{r_1,X,\varepsilon_1}$ and $(B,v)=G_{r_2,X,\varepsilon_2}$ with $r_1 + c <r_2$, $p\in S(A,u)$, $q\in S(B,v)$, such that $p_1(a_X,\pi)=0$ for all permutations $\pi$ and $q_1(b_X,\sigma)>0$ for some permutation $\sigma$. By Lemma~\ref{lm_lotteries_games_FOSD_undominated}, we have  $u_1((a_{r_1},\pi),p_2)+v_1((b_X,\sigma),q_2) = r_1+\varepsilon_1X+X$ and $u_1((a_X,\pi),p_2)+v_1((b_{r_2},\sigma),q_2) = X+r_2+\varepsilon_2X$. For $\varepsilon$ small enough, 
\[
u_1((a_{r_1},\pi),p_2)+v_1((b_X,\sigma),q_2) <_{\text{FOSD}} u_1((a_X,\pi),p_2)+v_1((b_{r_2},\sigma),q_2).
\]
However, since $p_1(a_X,\pi)=0$ for all $\pi$, there must be $\pi_0$ so that $p_1(a_{r_1},\pi_0)>0$. By bracketing, $p \times q \in S((A,u)\otimes (B,v))$. This violates distribution-monotonicity, since $p_1(a_{r_1},\pi_0)\cdot q_1(b_X,\sigma)>0$, while $p_1(a_X,\pi_0)\cdot q_1(b_{r_2},\sigma)=0$. Hence, the limit $\Phi[X]$ exists. 

\medskip
We next show $\Phi$ is a monotone additive statistic. 
Since players play ordinal-Nash, it is immediate that, for $c\in \R$, $\Phi_\varepsilon[c]=(1-\varepsilon)c$, so $\Phi[c]=c$; i.e., $\Phi$ is a statistic.

We need to show that $\Phi$ is additive for independent variables. We will show sub-additivity; super-additivity follows an identical argument. Let $X,Y  \in \Delta_\Q$ be independent, and let $r>\Phi[X],s>\Phi[Y]$. Fix $t >r+s$ and let $(A,u)=G_{r,X,\varepsilon}$, $(B,v)=G_{s,Y,\varepsilon}$, and $(C,w)=G_{t,X+Y,\varepsilon}$. Fix $o \in S(A,u)$, ${p}\in S(B,v)$, and ${q}\in S(C,w)$. By bracketing, $o\times p\times q\in S((A,u)\otimes(B,v)\otimes(C,w))$. Note that for any permutations $\pi$, $\sigma$ and $\tau$, by Lemma~\ref{lm_lotteries_games_FOSD_undominated},
\begin{align*}
    u_1((a_r,\pi),o_2)+u_1((a_s,\sigma),p_2)+u_1((a_{X+Y},\tau),q_2)&= X+Y+\varepsilon(X+Y)+r+s\\
    u_1((a_X,\pi),o_2)+u_1((a_Y,\sigma),p_2)+u_1((a_{t},\tau),q_{2})&= X+Y+\varepsilon(X+Y)+t,
\end{align*}
where the latter expression first-order dominates the former one. By distribution-monotonicity, we must have
\[o_1(a_r,\pi)\cdot p_1(a_s,\sigma)\cdot q_1(a_{X+Y},\tau) \leq o_1(a_X,\pi)\cdot p_1(a_Y,\sigma)\cdot q_1(a_t,\tau) .\]

From the definition of $\Phi$, for $\varepsilon$ small enough $o_1(a_X,\pi)=p_1(a_Y,\sigma)=0$. Since the above inequality must hold for all $\pi$ and $\tau$, we see that $q_1(a_{X+Y},\tau)=0$. Since we can choose $r$, $s$, and $t$ so that $t$ is arbitrarily close to $\Phi[X]+\Phi[Y]$, it follows that $\Phi[X+Y]\leq \Phi[X]+\Phi[Y]$.

We next show that $\Phi$ is monotone with respect to first-order stochastic dominance. Let $\varepsilon>0$, $X,Y\in \Delta_\Q$ with $X >_\text{FOSD} Y$, and fix $r < s$. Let $(A,u)=G_{r,X,\varepsilon}$, $(B,v)=G_{s,Y,\varepsilon}$, and fix ${p}\in S(A,u)$ and ${q}\in S(B,v)$. Note that, for all $\pi$ and $\sigma$ and for all $\varepsilon$ small enough, since $Y <_\text{FOSD} X$ and $r<s$,
\begin{align*}
 u_1((a_r,\pi),p_2)+v_1((b_Y,\sigma),q_2) 
 &= Y+ \varepsilon X + r   \\
 &<_\text{FOSD} X+ \varepsilon Y + s\\
 &= u_1((a_X,\pi),p_2)+v_1((b_s,\sigma),q_2).
\end{align*}
Thus, by distribution-monotonicity, ${p}_1(a_r,\pi)\cdot {q}_1(b_Y,\sigma)\leq {p}_1(a_X,\pi)\cdot {q}_1(b_s,\sigma)$. Hence, if ${p}_1(a_X,\pi)=0$ for all $\pi$, it must be that ${q}_1(b_Y,\sigma)=0$ for all $\sigma$. Recalling the definition of $\Phi_\varepsilon$ in \eqref{eq_phi_eps}, it follows that for any $r<s$, there is a $\delta>0$ such that for all $\varepsilon<\delta$, if $\Phi_\varepsilon[X] < r$ then $\Phi_\varepsilon[Y]\leq s$. Hence, $\Phi[X]\geq \Phi[Y]$. 

Finally, if $X$ and $Y$ have the same distribution then, for any $\varepsilon>0$, $\Phi[Y]-\varepsilon = \Phi[Y-\varepsilon] \leq\Phi[X]\leq \Phi[Y+\varepsilon]=\Phi[Y]+\varepsilon$, so $\Phi[X]=\Phi[Y]$.

\end{proof}

\begin{lemma}\label{lm_DSC_distribution}
Suppose a solution concept $S$ permits bracketing and is a refinement of ordinal-QRE. Then there is a monotone additive statistic $\Phi$ such that for all $G = (A,u), p \in S(G),$  and $a \in A_1$, $$p_1(a) \propto \exp(\lambda \Phi[u_1(a, p_{-1})])$$ for some $\lambda \geq 0$.
\end{lemma}
\begin{proof}[Proof of Lemma~\ref{lm_DSC_distribution}]

For each $X \in \Delta_\Q$, let $H_{0,X}=(B,v)$ be the game defined in \eqref{eq_HrX}, where $r=0$. We claim that since $S$ permits bracketing, $S(H_{0,X})$ is a singleton. 
Indeed, let $p,q \in S(H_{0,X})$, so $p\times q\in S(H_{0,X}\otimes H_{0,X})$ by  bracketing. Since \[v_1(b_X,p_{2})+v_1(b_0,q_{2})=X=v_1(b_0,p_{2})+v_1(b_X,q_{2}),\] it must be that $p_1(b_X)\cdot q_1(b_0)=p_1(b_0)\cdot q_1(b_X)$. Equivalently, $p_1(b_X)\cdot (1-q_1(b_X))=(1-p_1(b_X))\cdot q_1(b_X)$, so $p_1=q_1$. By distribution-neutrality, $p_i=q_i$ for $i\neq 1$, since they must both be the uniform distribution over $B_i$. Hence $p=q$, and $S(H_{0,X})$ is a singleton.

Define $f\colon \Delta_\Q\to \R$ by \[f(X) := \log \left(\frac{p_1(b_X)}{1-{p}_1(b_X)}\right),\] 
where $p$ is the unique element of $S(H_{0,X})$. This is finite by interiority. 

Let $X,Y \in \Delta_\Q$ be independent lotteries and let $o\in S(H_{0,X}), p \in S(H_{0,Y})$, and $q \in S(H_{0,X+Y})$. By bracketing, $o\times p\times q\in S(H_{0,X}\otimes H_{0,Y}\otimes H_{0,X+Y})$. By distribution-neutrality, $(1-o_1(b_{X}))\cdot(1-{p}_1(b_{Y}))\cdot q_1(b_{X+Y})=o_1(b_X)\cdot{p}_1(b_Y)\cdot(1-q_1(b_{X+Y}))$. Rearranging gives
\[
\frac{q_1(b_{X+Y})}{1-{q}_1(b_{X+Y})} =\frac{o_1(b_X)}{1-{o}_1(b_X)} \cdot \frac{p_1(b_Y)}{1-{p}_1(b_Y) }.
\]
Taking logs, we have $f(X + Y)=f(X)+f(Y)$; i.e., $f$ is additive for independent lotteries.

Since $S$ permits bracketing and satisfies distribution-monotonicity, for $X>_{\text{FOSD}} Y$, $p \in S(H_{0,X})$, $q \in S(H_{0,Y})$, we must have ${p}_1(b_X)\cdot {q_1}(b_0)\geq q_1(b_{Y})\cdot p_1(b_0)$; i.e., ${p}_1(b_X) \geq {p}_1(b_Y)$. Hence, $f$ is nondecreasing in first-order dominance. Finally, define $g \colon \R\to \R$ by  $g(x)=f(x)$ for deterministic lotteries yielding $x$ for sure. Then $g(x+y)=f(x+y)=f(x)+f(y)=g(x)+g(y)$. Since $f$ is monotone, $g$ must be monotone, so there is a $\lambda \in [0,\infty)$ such that $f(x)=g(x)=\lambda x$ for all $x \in \R$. Hence, $f$ is a monotone additive statistic on $\Delta_{\Q}$ scaled by $\lambda$. By Lemma~\ref{lm_rational_MAS}, $f(X)=\lambda \Phi[X]$, for some monotone additive statistic $\Phi$ on $\Delta_\fin$.

Fix any $G=(A,u), o\in S(G)$ with $a,a' \in A_1$. Let $X=u_1(a,o_{-1})$ and $Y=u_1(a',o_{-1})$. Fix $n\geq 1$, and let $(B,v) = H_{0,\bar X_n}$, $ (C,w) = H_{0,\ubar Y_n}$, and let $p \in S(H_{0,\bar X_n})$ and $q \in S(H_{0,\ubar Y_n})$. By bracketing, $o\times p\times q\in S(G \otimes H_{0,\bar X_n} \otimes H_{0,\ubar Y_n})$. By \eqref{eq_X_n_increasing}, 
\begin{align*}
u_1(a,o_{-1})+v_1(b_0,p_2)+w_1(c_{\ubar Y_n},q_{2}) 
&= X+\ubar Y_n\\
&<_\text{FOSD} \bar X_n+ Y\\
&=   u_1(a',o_{-1})+v_1(b_{\bar X_n},p_{2})+w_1(c_0,q_{2}).
\end{align*}
Thus, by distribution-monotonicity, $$o_1(a)\cdot(1-p_1(b_{\bar X_n}))\cdot q_1(c_{\ubar Y_n})\leq o_1(a')\cdot p_1(b_{\bar X_n})\cdot (1-q_1(c_{\ubar Y_n})).$$ By interiority, we have $$\frac{o_1(a)}{o_1(a')}\leq \frac{p_1(b_{\bar X_n})}{1-p_1(b_{\bar X_n})}\cdot\frac{1-q_1(c_{\ubar Y_n})}{q_1(c_{\ubar Y_n})} = \frac{\exp{f(\bar X_n)}}{\exp{f(\ubar Y_n)}}=\frac{\exp(\lambda \Phi[\bar X_n])}{\exp(\lambda \Phi[\ubar Y_n])}.$$ 

 By a symmetric argument, $$\frac{\exp(\lambda \Phi[\ubar X_n])}{\exp(\lambda \Phi[\bar Y_n])} \leq \frac{o_1(a)}{o_1(a')}.$$ Since these inequalities hold for all $n\geq 1$, by Lemma~\ref{lm_lim_Phi_converge}, $$\frac{o_1(a)}{o_1(a')}=\frac{\exp(\lambda \Phi[X])}{\exp(\lambda \Phi[Y])}.$$

 Since $a'$ was an arbitrary element of $A_1$, $$o_1(a) \propto \exp (\lambda \Phi[u_1(a,o_{-1})]). $$ 

\end{proof}
\begin{proof}[Proof of Theorem~\ref{th_sre}]
Since $S$ satisfies the assumptions of Proposition~\ref{lm_Ghl} proved above, $S$ is either a refinement of ordinal-Nash, or a refinement of ordinal-QRE. In the latter case, $S$ is a refinement of an $\LP$ by Lemma~\ref{lm_DSC_distribution}, and anonymity, which ensures that all players use the same $\lambda$ and $\Phi$. 

It remains to consider the case that $S$ is a refinement of ordinal-Nash. We show that in this case $S$ is a refinement of a $\Nash_\Phi$ equilibrium, concluding the proof of Theorem~\ref{th_sre}. Define the monotone additive statistic $\Phi \colon \Delta_\fin \to \R$ using Lemma~\ref{lm_phiep} to first define it on $\Delta_\Q$ and then applying the representation of Lemma~\ref{lm_rational_MAS} to obtain a monotone additive statistic on $\Delta_\fin$.

We will show that in any game only actions maximizing  $\Phi$ can be played with positive probability. Let $G=(A,u)$, with $a,b \in A_1, o\in S(G)$, and $\Phi[u_1(a,o_{-1})]>\Phi[u_1(b,o_{-1})]$. Let $X=u_1(a,o_{-1}),Y=u_1(b,o_{-1})$ so that $\Phi[Y] < \Phi[X]$, and denote $a_X=a$ and $a_Y=b$. 

By Lemma~\ref{lm_lim_Phi_converge}, there is an $n\geq 1$, such that $\Phi[\bar Y_n]<\Phi[\ubar X_n]$. Fix $r,s$ with $\Phi[\bar Y_n]<r<s < \Phi[\ubar X_n]$. 

Define 
\begin{align*}
    \Phi_\varepsilon[X]=\sup \{r\in \R \mid \exists p \in S(G_{r,X,\varepsilon}), \exists \pi \colon A_2 \to A_2 \text{ with } p_1(a_X,\pi)>0\},
\end{align*}
as in Lemma~\ref{lm_phiep}. By that lemma, there exists $\delta > 0$ such that for all $\varepsilon<\delta$, $\Phi_\varepsilon[\bar Y_n]<r<s < \Phi_\varepsilon[\ubar X_n]$. Fix $\varepsilon<\delta$, so that by the definition of $\Phi_\varepsilon$ there is a $t\in [s,\Phi_\varepsilon[\ubar X_n]]$, a $p \in S(G_{t,\ubar X_n,\varepsilon})$ and a $\pi$ such that $p_1(b_{\ubar X_n},\pi)>0$. Let $(B,v)=G_{t,\ubar X_n,\varepsilon}$ and $(C,w)=G_{r,\bar Y_n,\varepsilon}$. Let $q \in S(G_{r,\bar Y_n,\varepsilon})$, and note that for all $\sigma$, $q_1(c_{\bar Y_n},\sigma)=0$. For all $\varepsilon$ small enough,
\begin{align*}
    u_1&(a_Y,o_{-1})+v_1((b_{\ubar X_n},\pi),p_{-1})+w_1((c_r,\sigma),q_{-1})\\
    &=Y+\ubar X_n+r+\varepsilon\cdot \bar Y_n \\
    &<_\text{FOSD} X+t+\varepsilon\cdot \ubar X_n+\bar Y_n\\ &=u_1(a_X,o_{-1})+v_1((b_t,\pi),p_{-1})+w_1((c_{\bar Y_n},\sigma),q_{-1})),
\end{align*}
by \eqref{eq_X_n_increasing} and the fact that $r<t$. By bracketing, $o \times p \times q \in S((A,u)\otimes (B,v)\otimes (C,w))$. Thus, by distribution-monotonicity, \[o_1(a_Y) \cdot p_1(b_{\ubar X_n},\pi) \cdot q_1(c_r,\sigma) \leq o_1(a_X) \cdot p_1(b_t,\pi) \cdot q_1(c_{\bar Y_n},\sigma).\]

The right-hand side is zero, since $q_1(c_{\bar Y_n},\sigma)=0$ for all $\sigma$. There is therefore a $\sigma_0$ with $q_1(c_r,\sigma_0)>0$. Since $p_1(b_{\ubar X_n},\pi)>0$, it must be that $o_1(a_Y)=0$, demonstrating that only maximizers of $\Phi$ can be played with positive probability.

Since $S$ satisfies anonymity, we have now shown that if this case holds, there is a monotone additive statistic $\Phi$ such that for all games $G$, $p \in S(G)$, players~$i$, and $a_i \in A_i,$ \[\supp p_i\subseteq \argmax_{a_i} \Phi[u_i(a_i,p_{-i})].\]
\end{proof}

\section{Proof of Theorem \ref{th_complex}}\label{ap_C}
Since $S$ permits bracketing and satisfies distribution-monotonicity, anonymity and interiority, by Theorem~\ref{th_sre}, $S$ is a refinement of some $\LQRE_{\lambda \Phi}$. Scale invariance ensures that $\Phi$ belongs to the class of positively homogeneous monotone additive statistics, which we characterize in the following lemma. We use $\Delta_\fin$ to refer to the set of all lotteries with finite outcomes.

\begin{lemma}\label{lm_min_max_ex}
    Suppose that $\Phi\colon \Delta \to \R$ is a monotone additive statistic such that $\Phi[\beta X] = \beta\Phi[X]$ for all $X\in \Delta_\fin$ and some $\beta > 0$, $\beta \neq 1$. Then $\Phi$ is a convex combination of the  minimum, the  maximum, and the expectation.
\end{lemma}

\begin{proof}[Proof of Lemma~\ref{lm_min_max_ex}]
Let $\beta > 0$ and let $\Phi$ be a monotone additive statistic. By Lemma~\ref{lm_rational_MAS},  $\Phi[X] = \int K_a[X]\,\dd\mu(a)$. Then
\begin{align*}
  \Phi[\beta X]
  &= \int \frac{1}{a}\log\E[\ee^{a \beta X}]\,\dd\mu(a) \\
  &= \int \frac{\beta}{a \beta}\log\E[\ee^{a \beta X}]\,\dd\mu(a) \\
  &= \beta \int K_{a\beta}[X]\,\dd\mu(a) \\
  &=\beta\int K_{a}[X]\,\dd(\beta_*\mu)(a).
\end{align*}
Denote $\Psi[X] = \int K_a[X]\,\dd(\beta_*\mu)(a)$, and note that this is also a monotone additive statistic. Then $ \Phi[\beta X] = \beta\Psi[X]$.

Suppose $\Phi[\beta X] = \beta \Phi[X]$ for all $X$ and some $\beta > 0$. Hence $\beta\Phi[X] = \beta\Psi[X]$ for all $X$, and so $\Phi = \Psi$. By Lemma 5 of \cite*{mu2024monotone}  it follows that $\mu = \beta_*\mu$. Since a probability measure on $\R$ can only be invariant to rescaling by $\beta \neq 1$ if it is the point mass at $0$, it follows that $\mu(\{-\infty,+\infty,0\})=1$.
\end{proof}
It is straightforward to see that if $\mu$ is supported on $\{-\infty,+\infty,0\}$, then $\Phi$ satisfies $\Phi[\beta X]=\beta \Phi[X]$ for all $\beta > 0$. We proceed with the proof of Theorem~\ref{th_complex}.

\begin{proof}[Proof of Theorem~\ref{th_complex}]
    Since $S$ permits bracketing and satisfies distribution-monotonicity, anonymity and interiority, by Theorem~\ref{th_sre}, $S$ is a refinement of some $\LQRE_{\lambda \Phi}$. If $\lambda=0$, then $S$ is a refinement of $\text{Min-Max-Mean}_{(0,0,0)}$ and we are done, so consider $\lambda>0$. Let $X \in \Delta_\Q$ and consider the game $H_{r,X}=(B,v)$ where $r=\Phi[X]$ (see \eqref{eq_HrX}). Let $p \in S(H_{r,X})$ and note that $p_1(b_r)=p_1(b_X)$, so all players play uniform strategies. Hence, by scale invariance, $p \in S(B,\frac{1}{2}\cdot v)$.
    Since $S$ is a refinement of $\LQRE_{\lambda\Phi}$ with $\lambda>0$, it follows that $\Phi[\frac{1}{2} X]=\Phi[\frac{1}{2} r]=\frac{1}{2} r=\frac{1}{2} \Phi[X]$. By Lemma~\ref{lm_min_max_ex}, $\Phi$ is a convex combination of the minimum, the maximum and the expectation. The desired representation follows by setting $\lambda_1 = \lambda \mu(\{-\infty\})$, $\lambda_2 = \lambda \mu(\{0\})$ and $\lambda_3 = \lambda\mu(\{+\infty\})$.
\end{proof}

\section{Connections to \cite*{brandl2024axiomatic}}\label{sec_bnb}
\cite*{brandl2024axiomatic} characterize Nash equilibrium as the unique solution concept satisfying consistency, consequentialism, and rationality.
Their work focuses on how negligible changes in the strategic environment impact behavior, while our approach emphasizes how players frame multiple unrelated decisions. To further explore the relationship between the two perspectives, we now outline their axioms.
\begin{definition}\label{def_consistency} 
A solution concept $S$ is \emph{consistent} if for any $(A,u)$, $(A,v)$, and $\alpha \in [0,1]$
$$ S\big(A,u\big)\cap S\big(A,v\big)\subset S\big(A,\,\alpha u+(1-\alpha)v\big).$$
\end{definition}
Consistency requires that if a strategy profile is a solution to two games, it must 
also be a solution to any convex combination of them. Differently stated, given a mixed strategy profile $p$, the set of games $\{G \,:\, p \in S(G)\}$ that it solves must be convex for consistent~$S$.

Among SREs, only Nash equilibrium and $\LQRE_{\lambda}$ satisfy consistency.
While our Theorem~\ref{th_Nash QRE} characterizes these two solution concepts, consistency is not implied by the hypotheses of the theorem, as the theorem also allows for refinements. For example, trembling hand perfect equilibrium is a refinement of Nash equilibrium that permits   bracketing, satisfies expectation-monotonicity and anonymity, but violates consistency.\footnote{See $\S5$ of \cite*{brandl2024axiomatic} for an illustration of a consistency violation by trembling hand perfect equilibrium.}

To formulate the next axiom, we say that a game $G=(A,u)$ is a \emph{blow-up} of $H=(B,v)$ if there exist $f_i \colon A_i \to B_i$ such that $u_i(a)=v_i(f(a))$ for all players~$i$ and $a \in A$, where $f\colon A \to B$ is the map that applies $f_i$ to the $i$-th component of an action profile. Given a mixed strategy profile $p$ in $S(G)$, we denote by $q=f(p)$ the profile in $H$ given by $q_i(b_i) = p_i(f_i^{-1}(b_i))$.  

\begin{definition}\label{def_consequentialism}
A solution concept $S$ satisfies \emph{consequentialism} if for any games $H$ and $G$ such that $G$ is a blow-up of $H$, it holds that $p \in S(G)$ if and only if $f(p) \in S(H)$ for the witnessing $f$.
\end{definition}
Equivalently, consequentialism means that if two games $G,H$ are identical, except that $G$ contains an additional action $a_i'$ that is indistinguishable from $a_i$, then the solutions of the two games are the same, except that the probabilities assigned to $a_i$ in $H$ can be divided in any way between $a_i$ and $a_i'$ in $G$. In other words, 
duplicating actions should not affect behavior beyond splitting probabilities. Consequentialism is satisfied by every $\Nash_\Phi$ but is violated by every $\LP$.

\begin{definition}\label{def_rationality}
A solution concept $S$ satisfies \emph{rationality} if for any game $G = (A,u)$, player~$i$, and a strictly dominant action $a_i \in A_i$,  it holds that   $p_i(a_i) > 0$ for all $p \in S(G)$.
\end{definition}
That is, players play dominant strategies with positive probability. Distribution-monotonicity implies rationality and, in particular,
rationality is satisfied by all $\Nash_\Phi$ and~$\LP$.

\subsection{Implications of Consistency and Consequentialism}

The rationality assumption of Brandl and Brandt is straightforward, and so we focus on the connection between our axioms and their consequentialism and consistency. In particular, we investigate the connection between these properties and our scale-invariance and  bracketing.

First, we point out that consequentialism and consistency imply scale-invariance. In fact, they imply a much stronger property:  $S(A,u)\subset S(A,\alpha \cdot u )$ for any game $(A,u)$ and  all $\alpha \in (0,1)$.\footnote{This is in contrast with our scale-invariance axiom which only imposes that solutions which are uniform distributions be invariant to scaling down the payoff map.} Indeed, consequentialism implies that every mixed-strategy profile is a solution to the game $(A,0)$ whose payoff map vanishes for every action. The result then follows from consistency, since any scaled down game $(A,\alpha \cdot u)$ is the convex combination $(A, \alpha \cdot u + (1-\alpha) \cdot 0)$.

Consequentialism and consistency also imply a property that is a weakening of  bracketing: for any games $G$ and $H$ \textit{there exist} solutions $p\in S(G)$ and $ q\in S(H)$ such that $p\times q \in S(G\otimes H)$.  To see this, we first show that consequentialism and consistency imply that if $p \in S(G)$ and $q \in S(H)$ then $p\times q$ solves the 
game with the same action space as $G\otimes H$ but with payoffs halved. We write 
$\alpha \cdot W$ for a game with the same action set as $W$ and payoffs scaled by 
$\alpha$, so our goal is to show $p\times q\in S\left(\frac{1}{2}(G\otimes H)\right)$. Let $G=(A,u)$ and $ H=(B,v)$. Denote   $C=A\times B$ and 
define two auxiliary games $\hat{G}=(C, \hat u)$ and $\hat{H}=(C, \hat v)$ by
$$\hat u_i((a_1,b_1),\dots,(a_n,b_n))=u_i(a)\qquad\text{and}\qquad\hat v_i((a_1,b_1),\dots,(a_n,b_n))=v_i(b).$$
The game $\hat{G}$ is a blow-up of $G$ under the map $f$  that projects $A\times B$ to $A$; similarly, $\hat{H}$ is a blow-up of $H$ under the projection to $B$. 
Consequentialism implies that $p\times q \in S(\hat G)\cap S(\hat H)$. By consistency, $p\times q$ is also a solution to the game $\left(C,\frac{1}{2}\hat u + \frac{1}{2}\hat v\right)$ which is precisely $\frac{1}{2}(G\otimes H)$.
Finally, let $p\in S(2G)$ and $q \in S(2H)$, so that by consequentialism and consistency $p \in S(G)$, $q\in S(H)$, and $p \times q \in S(G\otimes H)$, as $G\otimes H= \frac{1}{2}\big((2G)\otimes (2H)\big)$.

Curiously,  bracketing is not implied by consequentialism and consistency: there are solution concepts $S$ satisfying these properties, with $p \in S(G)$ and $q \in S(H)$ such that $p \times q \not \in S(G \otimes H)$. For example, consider $\varepsilon$-Nash that assigns to a game $G$ all mixed strategy profiles $p$ such that $p_i(a_i) > 0$ implies $\E[u_i(a_i,p_{-i})] \geq \max_{b_i}\E[u_i(b_i,p_{-i})] - \varepsilon$. It is straightforward to check that $\varepsilon$-Nash satisfies consequentialism and consistency. However, $\varepsilon$-Nash violates  bracketing. Indeed, if player~$i$ plays an action that is $\varepsilon$-suboptimal in $G$ then  bracketing would imply that they play a $2\varepsilon$-suboptimal action in $G \otimes G$. Of course, by the main result of \cite*{brandl2024axiomatic},  bracketing is implied if we add their rationality assumption,  highlighting another connection between  bracketing and rationality.

\section{Strategic Invariance and the Emergence of Expected Utility}\label{app_strategic_eq}

Relaxing expectation-monotonicity to distribution-monotonicity yields new families of solution concepts, developed in \S\ref{sec:SRE}. That analysis demonstrates that distribution-monotonicity is much weaker than expectation-monotonicity, even when coupled with bracketing. In this section, we show that the gap between these axioms can be bridged with the additional assumption of \textit{strategic invariance}, which restricts a solution concept's predictions across \textit{strategically equivalent} games.

\begin{definition}\label{def_strat_eq_games}
    Games $(A,v)$ and $(A,u)$ are \emph{strategically equivalent} if for each player~$i$ there exists a function $w_i \colon A_{-i} \to \R$ such that $v_i(a) = u_i(a) + w_i(a_{-i})$.
\end{definition}
That is, strategically equivalent games share the same sets of actions, and while payoffs may differ, they must satisfy $v_i(a_i,a_{-i})-v_i(b_i,a_{-i}) = u_i(a_i,a_{-i})-u_i(b_i,a_{-i})$ for all $a_i,b_i \in A_i$ and $a_{-i} \in A_{-i}$. In other words, player~$i$'s marginal payoff of switching from an action $a_i$ to another action $b_i$ is the same in the two games.

The notion of strategic equivalence is fundamental to the study of solution concepts. For example, strategically equivalent games have identical sets of Nash and correlated equilibria. In mechanism design, strategic equivalence is an important tool. It provides the designer with the flexibility to modify a player~$i$'s transfers without altering their incentives, simply by adding a quantity that is independent of $i$'s report. This flexibility is crucial in mechanisms like VCG, where it helps to achieve the desired normalization of transfers and, in some environments, budget-balancedness.

\begin{definition}
    A solution concept $S$ satisfies \emph{strategic invariance} if  $S(A,u) = S(A,v)$ for strategically equivalent games $(A,u)$ and $(A,v)$.
\end{definition}
Strategic equivalence is respected by $\Nash$ and $\LQRE_\lambda$, as well as many other concepts that do not have a rational expectations component, such as rationalizability and level-$k$ reasoning.\footnote{One could imagine a more restrictive invariance notion that also includes invariance to rescaling of payoffs. We will not follow that route, since it would rule out $\LQRE$, which are not scale-invariant. \label{footnote:scale-invariance}}

Our next theorem shows that strategic invariance---when coupled with  bracketing---becomes a powerful assumption that elevates distribution-monotonicity to expectation-monotonicity. 
\begin{theorem}\label{th_strategic_equiv}
Suppose $S$ permits bracketing and satisfies strategic invariance, distribution-monotonicity, and anonymity. Then it satisfies expectation-monotonicity.
\end{theorem}

Recall that distribution-monotonicity is a rational expectations and monotonicity axiom, and does not have an expected utility or risk neutrality component. In contrast, expectation-monotonicity is a stronger axiom that implies distribution-monotonicity, and furthermore also has expected utility or risk neutrality components. Theorem~\ref{th_strategic_equiv} thus shows that strategic invariance is a potent assumption that highly constrains behavior to resemble risk neutrality. 

Combining Theorems~\ref{th_Nash QRE} and \ref{th_strategic_equiv}, we obtain the following corollary, which offers a foundation for Nash and $\LQRE_\lambda$ without directly assuming that behavior is driven by expectations.
\begin{corollary}
    Suppose $S$ permits bracketing and satisfies strategic invariance, distribution-monotonicity, and anonymity, then $S$ is either a refinement of $\Nash$ or of $\LQRE_\lambda$ for some $\lambda\geq 0$.
\end{corollary}

The proof of Theorem~\ref{th_strategic_equiv} is provided below. To build intuition, it is helpful to interpret payoffs as monetary. Under this interpretation, our aim is to understand how strategic invariance, together with the other axioms, rules out non-trivial risk attitudes.

Strategic invariance immediately implies that players display no wealth effects, since adding a constant to all payoffs does not change their behavior. To see why strategic invariance furthermore rules out any non-trivial risk attitudes, consider the following example of a two-player game. Player $2$ has two actions, $a_2$ and $b_2$, and gets payoff $0$ regardless of the action profile. For now, assume that this player mixes evenly between these two actions. Player $1$ has two actions, $a_1$ and $b_1$, and gets the payoffs presented on the left side of Table~\ref{tab:game_1}.
\begin{table}[h]
    \centering
    \begin{tabular}{c|c|c}
            & $a_2$&$b_2$ \\ \hline
         $a_1$&  $0$& $2$\\
         $b_1$&  $1$& $1$\\
    \end{tabular}
    \qquad 
    \qquad
    \begin{tabular}{c|c|c}
            & $a_2$&$b_2$ \\ \hline
         $a_1$&  $0$& $1$\\
         $b_1$&  $1$& $0$\\
    \end{tabular}
    \caption{Player $1$'s payoffs in two strategically equivalent games.}
    \label{tab:game_1}
\end{table}\\
In this game, both actions yield the same expected payoff, but action $a_1$ has variance $1$, whereas action $b_1$ has variance $0$, and so would be preferred by any risk-averse player. Consider now the strategically equivalent game described on the right side of Table~\ref{tab:game_1}.
Here, both actions yield the same distribution of payoffs to player~$1$, and hence risk attitudes should not influence the choice between $a_1$ and $b_1$. Since this game is strategically equivalent to the previous, we conclude that under strategic invariance, players would be indifferent between the two actions in the previous game and so are effectively risk-neutral.

The assumption that player~2 mixes evenly between the two actions is crucial for this argument and turns out to be non-trivial: if we cannot guarantee mixing by player~2, we cannot conclude anything about the risk attitudes of player~1. But in these games, player~2 has no particular reason to mix, and so the actual proof of Theorem~\ref{th_strategic_equiv} relies on the test games that we developed for the proof of Theorem~\ref{th_sre} (see Appendix~\ref{ap_test}).

 In these games, player~2 mixes, generating a choice between a sure thing and a lottery for player~1. By applying strategic invariance and the idea behind the simple games in Table~\ref{tab:game_1}, we show that player~1 will choose the sure thing if it is higher than the expectation of the lottery. To prove Theorem~\ref{th_strategic_equiv}, we need to extend this conclusion to all games. This step relies on Theorem~\ref{th_sre}.

\bigskip
While strategic invariance is a strong assumption that is suggestive of expected-utility maximization, it does not imply expected-utility maximization without the additional assumption of bracketing. To see this, we construct a solution concept for two players that satisfies strategic invariance as well as distribution-monotonicity, but does not satisfy expectation-monotonicity. This shows that bracketing is an important component in achieving expectation-monotonicity.

Consider the following family of games parameterized by $w,x,y,z\in \R$. Note that these games form an equivalence class under strategic equivalence. We suppose that for this class of games, 
the solution concept $S$ is the singleton solution in which player~1 chooses $a_1$ with probability $2/3$ and $b_1$ with probability $1/3$, and player~2 chooses $a_2$ with probability $1/3$ and $b_2$ with probability $2/3$. Thus, $S$ violates expectation-monotonicity, as  player~1's expected payoff is higher for $b_1$ than $a_1$.

 \begin{table}[h!]
    \centering
  \begin{tabular}{c|c|c}
            & $a_2$&$b_2$ \\ \hline
         $a_1$&  $(1.5+w,y)$& $(x,1+y)$\\
         $b_1$&  $(w,1.5+z)$& $(1+x,z)$\\
    \end{tabular}
    \caption{Equivalence class of games}
    \label{tab:counterexample}
\end{table}

It is easy to see that neither player faces a choice between strictly FOSD-ordered actions in any of these games. 
 Hence neither strategic invariance nor distribution-monotonicity is violated here. Finally, one can set $S$ to be (say) the Nash solution on all other strategic-equivalence classes; since we are not assuming bracketing, there is no interaction across classes. Thus $S$ violates expectation-monotonicity, while satisfying distribution-monotonicity and strategic invariance.

This example highlights the importance of bracketing in connecting strategic invariance to expectation-monotonicity. Intuitively, there is a link between strategic invariance and expectation-monotonicity. However, this intuition requires agents to have a preference over lotteries, which is not assumed, and is achieved in our setting by the addition of bracketing.

 \begin{proof}[Proof of Theorem~\ref{th_strategic_equiv}]

Let $X,Y \in \Delta_{\mathbb Q}$. As previously, we represent $X$ and $Y$ as random variables on $(\Omega=\{1,\dots,m\},2^\Omega,\nu)$, where $\nu$ is uniform. We consider the two families of SREs characterized by Theorem~\ref{th_sre}. First, we consider $S$ that is a refinement of a $\Nash_\Phi$ equilibrium. 

Let $r<\Phi[X]$ and $\varepsilon\in \left(0,\frac{\Phi[X]-r}{\max|X|+1}\right)$, and let $G_{r,X,\varepsilon}=(A,u)$ be the game defined in~\eqref{eq:G_rXeps}.\footnote{As in Appendix~\ref{ap_test}, we let $G_{r,X,\varepsilon}$ be equal to some $G_{r,x,\varepsilon}$ for $x = (x_1,\ldots,x_m)$ such that the uniform distribution over $\{x_1,\ldots,x_m\}$ is $X$.} We will consider the probability that player~$1$ chooses the (almost) sure things $(a_r,\,\cdot\,)$. Note that Lemma~\ref{lm_lotteries_games_FOSD_undominated} applies since players never play first-order dominated actions in any $\Nash_{\Phi}$ equilibrium. We will show that $\Phi$ is additive for all lotteries (rather than only independent ones).

We first show super-additivity. Let $p \in S(G_{r,X,\varepsilon})$. Since $p$ is a $\Nash_\Phi$ equilibrium and \[\Phi[r+\varepsilon X] \leq r + \Phi[\varepsilon \max|X|] \leq r+\Phi\left[(\Phi[X]-r)\frac{\max|X|}{\max|X|+1}\right]<r+\Phi[X]-r=\Phi[X],\] $p_1(a_r,\pi)=0$ for all $\pi$, and $p_1(a_X,\sigma)>0$ for some $\sigma$. Consider the game $(A,v),$ where $v_1(a_1,a_2)=u_1(a_1,a_2) + Y(\sigma(a_2))$ for each $(a_1,a_2) \in A_1 \times A_2$, and $v_i=u_i$ for $i\neq 1$. By strategic invariance $p\in S(A,v)$, and by Lemma~\ref{lm_lotteries_games_FOSD_undominated}, $p_2$ is the uniform distribution. 

Note that $v_1((a_X,\sigma),p_2)$ is distributed as $X+Y$, and $v_1((a_r,\sigma),p_2)$ is distributed as $r + Y + \varepsilon X$. Since $p_1((a_r,\sigma))=0$ and $p_1((a_X,\sigma)) > 0$, it must be that $\Phi[r+Y+\varepsilon X]=r+\Phi[Y+\varepsilon X] \leq \Phi[X+Y]$. Taking $\varepsilon \to 0$, we see that $r + \Phi[Y] \leq \Phi[X+Y]$.\footnote{Note that $\Phi[Y]+\varepsilon \min X = \Phi[Y+\varepsilon \min X] \leq \Phi[Y+\varepsilon X] \leq \Phi[Y+\varepsilon \max X]=\Phi[Y]+\varepsilon \max X$, so $\lim_{\varepsilon\to 0}\Phi[Y+\varepsilon X]=\Phi[Y]$.} Since $r<\Phi[X]$ was arbitrary, it follows that $\Phi[X]+\Phi[Y]\leq \Phi[X+Y]$.

We next show sub-additivity. Let $r>\Phi[X]$ and $\varepsilon\in \left(0,\frac{r-\Phi[X]}{\max|X|+1}\right)$, and let $G_{r,X,\varepsilon}=(A,u)$ be the game defined in \eqref{eq:G_rXeps}. Let $p \in S(G_{r,X,\varepsilon})$. Since $p$ is a $\Nash_\Phi$ equilibrium and
\begin{align*}
\Phi[r+\varepsilon X] = r + \Phi[\varepsilon X] & \geq r+\Phi[-\varepsilon \max|X|] \\  & \geq r+\Phi\left[-(r-\Phi[X])\frac{\max|X|}{\max|X|+1}\right]>r-r+\Phi[X]=\Phi[X],    
\end{align*}
$p_1(a_X,\pi)=0$ for all $\pi$, and $p_1(a_r,\sigma)>0$ for some $\sigma$. Consider the game $(A,v),$ where $v_1(a_1,a_2)=u_1(a_1,a_2) + Y(\sigma(a_2))$ for each $(a_1,a_2) \in A_1 \times A_2$, and $v_i=u_i$ for $i\neq 1$. By strategic invariance, $p \in S(A,v)$ and by Lemma~\ref{lm_lotteries_games_FOSD_undominated}, $p_2$ is the uniform distribution. 

Note that $v_1((a_X,\sigma),p_2)$ is distributed as $X+Y$, and $v_1((a_r,\sigma),p_2)$ is distributed as $r + Y + \varepsilon X$. Since $p_1((a_X,\sigma))=0$ and $p_1(a_r,\sigma) > 0$, it must be that $\Phi[r+Y+\varepsilon X]=r+\Phi[Y+\varepsilon X] \geq \Phi[X+Y]$. Then by an analogous argument for sub-additivity, we have
$\Phi[X]+\Phi[Y] \geq \Phi[X+Y]$. Thus $\Phi[X+Y]=\Phi[X]+\Phi[Y]$.
By Lemma~\ref{lm_additive_MAS}, $\Phi$ is the expectation on $\Delta_{\mathbb Q}$, so by Lemma~\ref{lm_rational_MAS}, $\Phi$ is the expectation.

\medskip
We next consider the case where $S$ is a refinement of some $\LP$ equilibrium. If $\lambda=0$, the result holds trivially. Consider then $\lambda > 0$. Note we can apply Lemma~\ref{lm_lotteries_games} since $\LP$ equilibrium satisfies distribution-neutrality. Given the game $H_{r,X}=(B,v)$ defined in \eqref{eq_HrX} with $r=\Phi[X]$, we define a new game $(B,w)$ by $w_1(b_1,b_2)=v_1(b_1,b_2)+Y(b_2)$, and $w_i=v_i$ for $i\neq 1$. Let $p \in S(B,w)$. By strategic invariance $p \in S(B,v)$, so $p_2$ is uniform by Lemma~\ref{lm_lotteries_games} and $w_1(b_r,p_2)$ is distributed as $r+Y$, while $w_1(b_X,p_2)$ is distributed as $X+Y$. 

Since $r=
\Phi[X]$ and $p$ is an $\LP$ equilibrium, we have \[p_1(b_r) = p_1(b_X) \ra \Phi[X+Y]=\Phi[r+Y]=\Phi[Y]+r=\Phi[X]+\Phi[Y].\]

Thus $\Phi$ is additive for all lotteries and is therefore the expectation by Lemma~\ref{lm_additive_MAS}. 

\end{proof}

\section{Non-Expected-Utility Properties of Statistic Response  Equilibria}\label{app_CARA}

In this section we show that SREs are incompatible with expected-utility maximization, except for the particular case in which the statistic $\Phi$ is of the form $K_a[X]=\frac{1}{a}\log \E\left[\ee^{aX}\right]$ for some $a \in \R$.

\begin{definition}
   A solution concept $S$ \emph{is compatible with expected utility} if 
   there exists a strictly increasing continuous function $f\colon \R\to \R$ such that $$\E \big[f\big(u_i(a_i,p_{-i})\big)\big]>\E \big[f\big(u_i(b_i,p_{-i})\big)\big]\qquad \text{implies}\qquad p_i(a_i)\geq  p_i(b_i)$$ for every game  $G=(A,u)$, solution $p\in S(G)$, player~$i$, and actions $a_i,b_i\in A_i$.
\end{definition}
That is, $S$ is compatible with expected utility if it satisfies expectation-monotonicity after a monotone reparameterization of payoffs.
\begin{proposition}\label{th_reparam}
An SRE  is compatible with expected utility if and only if it is either $\Nash_{K_a}$ or $\LQRE_{\lambda K_a}$ for some $a \in \R$ and~$\lambda \geq 0$.    
\end{proposition}

\begin{proof}[Proof of Proposition~\ref{th_reparam}]
An SRE $S$ is either $\Nash_{\Phi}$ or $\LQRE_{\lambda\Phi}$ for some monotone additive statistic $\Phi$ and $\lambda \geq 0$. For $\Phi = K_a$ and $a \in \R$, it is immediate that $S$ is compatible with expected utility, since we can take $f$ to be the CARA utility. 
 
It is left to show that if $S$ is compatible with expected utility, then $S$ is either $\Nash_{K_a}$ or $\LQRE_{\lambda K_a}$. For $\LQRE_{\lambda\Phi}$ where $\lambda = 0$, the result is trivial. Otherwise, suppose $f \colon \R \to \R$ witnesses that $S$ is compatible with expected utility. Let $S$ have statistic $\Phi=\int K_a\, \dd \mu(a)$. We define the statistic $\Psi$ by $\Psi[X]=f^{-1}(\E[f(X)])$. For any game $G$ and $p\in S(G)$, $\argmax_a p_1(a) \cap \argmax_a \Psi[u_1(a,p_{-1})] \neq \emptyset$. Moreover, since $S$ is an SRE with statistic $\Phi$, $\argmax_a p_i(a) \subset \argmax_a \Phi[u_1(a,p_{-1})]$. Thus there is an action $a \in A_1$ that maximizes both $\Phi$ and $\Psi$ under $p_{-1}$. 

Let $X \in \Delta_\Q$ be the uniform distribution over the coordinates of $x \in \R^m$ and fix $r \in \R$ and $\varepsilon>0$. Suppose first that $S=\Nash_\Phi$ and let $p\in S(G_{r,x,\varepsilon})$. By Lemma~\ref{lm_lotteries_games_FOSD_undominated}, $p_2$ is the uniform distribution, and so an action of the form $(a_x,\pi)$ maximizes $\Phi$ under $p_2$ if and only if $\Phi[X] \ge \Phi[r+\varepsilon X]$, and likewise for $\Psi$. Since there is an action that maximizes both $\Phi$ and $\Psi$, by taking $\varepsilon\to 0$, we see that $\Phi[X] \ge r$ if and only if $\Psi[X] \ge r$, i.e. $\Phi|_{\Delta_\Q}=\Psi|_{\Delta_\Q}$.

Likewise, if $S=\LP$ for some $\lambda>0$, we consider any $p\in S(H_{r,x})$ so that $u_1(b_x,p_2)=X$ by Lemma~\ref{lm_lotteries_games} . Thus $b_x$ maximizes $\Phi$ if and only if $\Phi[X] \ge r$, and likewise for $\Psi$. Since there is an action that maximizes both $\Phi$ and $\Psi$, we again conclude that $\Phi|_{\Delta_\Q}=\Psi|_{\Delta_\Q}$.

By Lemma~\ref{lm_rational_MAS}, there is a statistic on the set of compactly supported lotteries that is monotone with respect to first-order dominance and coincides with $\Phi$ on $\Delta_\Q$. Moreover, by \eqref{eq_X_n_increasing} and Lemma~\ref{lm_lim_Phi_converge} this extension is unique. We thus conclude that $\Psi[X]=\Phi[X]$ for all compactly supported $X$. 

Since $f \circ \Psi(X) = \E[f(X)]$, $\Psi$ satisfies independence, i.e., for all compactly supported lotteries $X,Y,Z$ and all $\beta \in (0,1)$, $\Psi[X] \ge \Psi[Y]$ if and only if  $\Psi[X_\beta Z] \ge \Psi[Y_\beta Z]$, where $X_\beta Z$ denotes the compound lottery that equals $X$ with probability $\beta$ and $Z$ with probability $1-\beta$, and likewise for $Y_\beta Z$. Hence, by Proposition~F.1 of \cite*{mu2024monotone}, $\Psi[X]=K_a[X]$ for $a \in \R$.\footnote{\cite*{mu2024monotone} consider a weakening of the independence axiom: $\Phi[X] \ge \Phi[Y] \implies \Phi[X_{\beta} Z] \ge \Phi[Y_{\beta} Z]$. They show that if $\Phi$ is a monotone additive statistic that satisfies this property then  $\Phi=K_a$ for $a\in \overline \R$. The stronger independence axiom rules out $K_{-\infty}$ and $K_{+\infty}$ since if $a\in \{-\infty,\infty\}$ and $K_a(Z)>K_a(Y)>K_a(X)$ then $K_a(X_{\beta}Z)=K_a(Y_{\beta}Z)$.}
\end{proof}

By this result, in $\LQRE_{\lambda K_a}$ players exhibit behavior that is compatible with expected utility maximization. To see this directly, we note that their choice probabilities are given by
\begin{align*}
    p_i(a_i) \propto  
\ee^{\frac{\lambda}{a}\log\E\big[\exp(a \cdot u_i(a_i,p_{-i}))\big]}.
\end{align*}
Note that this is not the same as  logit responding to the transformed payoffs, as in \cite*{goeree2003risk}, which would require
\begin{align*}
    p_i(a_i) \propto \ee^{\frac{\lambda}{a}\E\big[\exp(a \cdot u_i(a_i,p_{-i}))\big]}.
\end{align*}
Their goal was to introduce risk averse behavior to $\LQRE_\lambda$. We conclude that this cannot be achieved via payoff transformations without giving up on bracketing, while $\LQRE_{\lambda K_a}$ can incorporate risk attitudes while maintaining bracketing.  

\section{Additivity as a Notion of Separability}\label{app_generalized_product_games}
 
This appendix makes precise the utility interpretation of payoffs outlined in \S\ref{sec_model}. In this interpretation, payoffs are utilities representing players' ordinal preferences over an underlying space of physical outcomes that games can produce. The challenge is to understand why these utilities are added in composite games, as addition is not an obviously meaningful operation to apply to utilities. 

In our model, we define games using payoffs associated with strategy profiles. We now discuss the origin of these payoffs. Consider, for example, an experimenter who designs a game. In reality, the experimenter specifies the physical outcomes of players' actions, and has no direct control over their preferences. Accordingly, in the background of our model there is a space $\cO$ of \emph{physical outcomes} (for instance,  prizes, allocations, experiences). We assume that $\cO$ is a connected topological space that admits a dense countable subset. For each player~$i$, there is a continuous preference relation $\succsim_i$ over $\cO$; we assume that $\succ_i$ is non-trivial, meaning that $s \succ_i t$ for some $s,t \in \cO$.
Each game is induced by a \emph{game form}: a finite set of action profiles $A = \prod_i A_i$ together with an outcome map $g\colon A \to \cO$ \cite[see, e.g.,][]{osborne1998games}. The game form captures how players' actions translate into physical outcomes.

Separability of strategic interactions can be modeled at the level of physical outcomes. For this purpose, we endow $\cO$  with a continuous, associative \emph{juxtaposition} operation~$\oplus$. We interpret  $o \oplus o'$ as the outcome of receiving $o$ in one interaction and $o'$ in a separate, unrelated one. This interpretation is formally captured by the  following \emph{separability} assumption: \begin{align}\label{eq_separability}
o_1 \succsim_i o_2 \ \iff \  o_1 \oplus o \succsim_i o_2\oplus o,\ \ \  \text{ and }\ \ \ o_1\succsim_i o_2 \ \iff \  o\oplus o_1\succsim_i o\oplus o_2
\end{align}
for all  outcomes $o_1,o_2,o\in\cO$ and players $i$. Separability excludes wealth effects across interactions, under which the ranking of $o_1$ and $o_2$ could change after both are juxtaposed with the same additional outcome.

Two game forms $(A,g)$ and $(B,h)$ played simultaneously can be combined into the composite game form $(A\times B,\, g\oplus h)$, where $(g \oplus h)(a,b) = g(a)\oplus h(b)$: each action profile of the composite interaction produces the juxtaposition of the outcomes of its components. Assumption~\eqref{eq_separability}  implies that the composite game captures the separability of the interactions. We require juxtaposition to be associative so that a three-fold composite is unambiguous: composing three game forms in either grouping yields the same game form.

To get a game with numerical payoffs from a game form, one needs to choose a cardinal representation $u_i \colon \cO \to \R$ of each preference $\succsim_i$. We show that the separability assumption~\eqref{eq_separability} implies that a preference $\succsim_i$ admits a distinguished representation in which the payoffs have cardinal significance: adding them corresponds to combining outcomes via juxtaposition. 
\begin{proposition}
\label{prop:additive}
There is a continuous representation $u_i\colon\cO\to\R$ of $\succsim_i$ such that
\begin{align*}
    u_i(o\oplus o') = u_i(o)+u_i(o') \qquad \text{for all } o,o' \in \cO.
\end{align*}
This representation is unique up to multiplication by a positive constant.
\end{proposition}
We call the representation $u_i$ the \emph{additive cardinalization} of $\succsim_i$. In contrast to von Neumann-Morgenstern utilities, this cardinalization is not driven by preferences over lotteries. Our representation is pinned down, up to scale, by how outcomes combine across unrelated interactions. The connection between separability and additivity established by this proposition is in line with the results of \cite*{debreu1959topological} and \cite*{gorman1968structure}, where separability across the dimensions of a product space plays the role of separability across juxtapositions.
 
This proposition justifies our definition of a composite game. Together with a profile of additive cardinalizations $(u_i)_i$, a game form $(A,g)$ becomes the game $(A,w)$ with payoffs $w_i(a)=u_i(g(a))$, and the composite game form $(A\times B,\,g\oplus h)$ becomes precisely the additive composite game $(A,u\circ g)\otimes(B,u\circ h)$. Additivity of payoffs is therefore not arbitrary but the expression, in the additive cardinalization, of the separability of the underlying preferences over juxtaposed outcomes.\footnote{The proposition does not imply that every payoff vector in $\R^n$ is induced by some outcome in $\cO$. Indeed, if all players share the same preference, their additive cardinalizations are proportional, so $u(\cO)$ lies on a line. Recovering the unrestricted payoff domain requires further assumptions, and the following pair is one example. The first is a cross-player richness condition: for every player $i$ and all $s,t\in\cO$, there exists $o\in\cO$ with $o\sim_i s$ and $o\sim_j t$ for every $j\ne i$; iterating it across players yields $u(\cO)=\prod_{i}u_i(\cO)$. The second is compensability: for every $i$ and all $s,t\in\cO$, there exists $r\in\cO$ with $r\oplus s\sim_i t$. It yields $u_i(\cO)=\R$ for every $i$ and hence $u(\cO)=\R^n$.}

Proposition~\ref{prop:additive} provides an ordinal foundation for the bracketing axiom, by interpreting payoffs as additive cardinalizations of the underlying preferences. 
Theorem~\ref{th_sre} shows that together with another ordinal axiom---distribution-monotonicity---this same cardinalization  organizes players' attitudes toward the endogenous randomness they face in equilibrium: the statistic $\Phi$ is computed from payoffs measured on precisely this scale.

\begin{proof}[Proof of Proposition~\ref{prop:additive}]
Consider a continuous representation $v_i$ which exists by \cite{debreu1953representation}.
By continuity, connectedness of $\cO$, and non-triviality of $\succsim_i$, the image $J=v_i(\cO)$ is an interval. Define a binary operation $\boxplus$ on $J$ by setting, for $s,t\in J$,
\begin{align*}
    s\boxplus t := v_i(o\oplus o'),\qquad\text{for any }o,o'\in\cO\text{ with }v_i(o)=s\text{ and }v_i(o')=t.
\end{align*}
We argue that this is well defined: if $v_i(o)=v_i(\tilde o)$ and $v_i(o')=v_i(\tilde o')$, that is, $o\sim_i\tilde o$ and $o'\sim_i\tilde o'$, then the first equivalence in~\eqref{eq_separability} gives $o\oplus o'\sim_i\tilde o\oplus o'$ and the second gives $\tilde o\oplus o'\sim_i\tilde o\oplus\tilde o'$, so by transitivity $v_i(o\oplus o')=v_i(\tilde o\oplus\tilde o')$, a value in $J$ since $o\oplus o'\in\cO$. The same two equivalences imply that $\boxplus$ is strictly increasing in each argument; and the associativity of $\oplus$ makes $\boxplus$ associative, since with $s=v_i(o)$, $t=v_i(o')$, $r=v_i(o'')$ one has $(s\boxplus t)\boxplus r=v_i\big((o\oplus o')\oplus o''\big)=v_i\big(o\oplus(o'\oplus o'')\big)=s\boxplus(t\boxplus r)$.
Since $v_i$ and $\oplus$ are continuous, so is $\boxplus$.
 
Hence $\boxplus$ is an associative continuous operation on the interval $J$ that is strictly increasing in each argument. By the theorem of \cite*{aczel1948operations} on such operations, there is a continuous and strictly increasing $\varphi\colon J\to\R$ with $s\boxplus t=\varphi^{-1}\big(\varphi(s)+\varphi(t)\big)$ for all $s,t\in J$. Set $u_i:=\varphi\circ v_i$. Then $u_i$ is continuous, represents $\succsim_i$, and is additive:
\begin{align*}
    u_i(o\oplus o')=\varphi\big(v_i(o)\boxplus v_i(o')\big)=\varphi\big(\varphi^{-1}(\varphi(v_i(o))+\varphi(v_i(o')))\big)=u_i(o)+u_i(o').
\end{align*}
Finally, we argue that $u_i$ is unique, up to multiplication by a positive scalar. Indeed, let $\hat u_i$ be another additive cardinalization of $\succsim_i$. Denote the interval $u_i(\cO)$ by $I$ and define $\psi\colon I\to \R$ by $\psi(t)=\hat u_i(o)$ for any $o$ with $u_i(o)=t$. This mapping is well defined, since $u_i(o)=u_i(o')$ implies $o\sim_i o'$ and hence $\hat u_i(o)=\hat u_i(o')$. A similar argument shows that $\psi$ is strictly increasing: $u_i(o)<u_i(o')$ implies $o\prec_i o'$ and hence $\hat u_i(o)<\hat u_i(o')$. Additivity of $u_i$ and $\hat u_i$ gives the Cauchy equation $\psi(x+y)=\psi(x)+\psi(y)$, whose monotone solutions are linear. Thus $\psi(x)=c\,x$ with $c>0$, and so $\hat u_i=c\,u_i$.
\end{proof}

In our definition of a composite game, we add the payoffs of the component games. Proposition~\ref{prop:additive} provides a foundation for this modeling choice through the presence of a physical outcome space endowed with a  juxtaposition operation. We complement this foundation by showing that any alternative operation of combining games, defined directly on payoffs, can be reduced to addition after some reparameterization.  
 
Consider the following alternative definition of the composite game $(C,w)=(A,u) \otimes (B,v)$ in which $C=A\times B$ and
\begin{align*}
    w_i(a,b) = f\Big(u_i(a),\,v_i(b)\Big),
\end{align*}
for some $f \colon \R^2 \to \R$. That is, player~$i$'s payoff in a composite game is determined by payoffs in component games but the dependence can be general. To make this a reasonable model, suppose $f$ is continuous, associative and strictly increasing.

By \cite{aczel1948operations} (which is the main ingredient in the proof of Proposition~\ref{prop:additive}), all such $f$ are of the form
$$f(x,y)=\varphi^{-1}(\varphi(x)+\varphi(y))$$
for some strictly increasing continuous $\varphi \colon \R \to \R$. In other words, up to an order-preserving reparameterization of the reals, $f$ is just addition.

\end{document}